\documentclass[twoside,a4paper,11pt]{article}
\usepackage{amsmath}
\usepackage{amsthm}
\usepackage{amssymb}
\usepackage{amscd}
\usepackage{graphicx}
\usepackage{epsfig}
\usepackage{bm}

\usepackage{latexsym}
\usepackage{amsfonts}
\usepackage{color}

\input xy
\xyoption{all}

-\headheight\advance\topmargin by -\headsep

\font\black=cmbx10 \font\sblack=cmbx7 \font\ssblack=cmbx5 \font\blackital=cmmib10  \skewchar\blackital='177
\font\sblackital=cmmib7 \skewchar\sblackital='177 \font\ssblackital=cmmib5 \skewchar\ssblackital='177
\font\sanss=cmss11 \font\ssanss=cmss8 
\font\sssanss=cmss8 scaled 600 \font\blackboard=msbm10 \font\sblackboard=msbm7 \font\ssblackboard=msbm5
\font\caligr=eusm10 \font\scaligr=eusm7 \font\sscaligr=eusm5  \font\fraktur=eufm10
\font\sfraktur=eufm7 \font\ssfraktur=eufm5

\font\bsymb=cmsy10 scaled\magstep2
\def\all#1{\setbox0=\hbox{\lower1.5pt\hbox{\bsymb
       \char"38}}\setbox1=\hbox{$_{#1}$} \box0\lower2pt\box1\;}
\def\exi#1{\setbox0=\hbox{\lower1.5pt\hbox{\bsymb \char"39}}
       \setbox1=\hbox{$_{#1}$} \box0\lower2pt\box1\;}

\def\tx#1{{\fam0\relax#1}}

\newfam\bifam
\textfont\bifam=\blackital \scriptfont\bifam=\sblackital \scriptscriptfont\bifam=\ssblackital
\def\bi#1{{\fam\bifam\relax#1}}

\newfam\blfam
\textfont\blfam=\black \scriptfont\blfam=\sblack \scriptscriptfont\blfam=\ssblack

\newfam\bbfam
\textfont\bbfam=\blackboard \scriptfont\bbfam=\sblackboard \scriptscriptfont\bbfam=\ssblackboard

\newfam\ssfam
\textfont\ssfam=\sanss \scriptfont\ssfam=\ssanss \scriptscriptfont\ssfam=\sssanss
\def\sss#1{{\fam\ssfam\relax#1}}

\newfam\clfam
\textfont\clfam=\caligr \scriptfont\clfam=\scaligr \scriptscriptfont\clfam=\sscaligr

\newfam\frfam
\textfont\frfam=\fraktur \scriptfont\frfam=\sfraktur \scriptscriptfont\frfam=\ssfraktur

\def\hpb#1{\setbox0=\hbox{${#1}$}
    \copy0 \kern-\wd0 \kern.2pt \box0}
\def\vpb#1{\setbox0=\hbox{${#1}$}
    \copy0 \kern-\wd0 \raise.08pt \box0}

\def\pmb#1{\setbox0\hbox{${#1}$} \copy0 \kern-\wd0 \kern.2pt \box0}
\def\pmbb#1{\setbox0\hbox{${#1}$} \copy0 \kern-\wd0
      \kern.2pt \copy0 \kern-\wd0 \kern.2pt \box0}
\def\pmbbb#1{\setbox0\hbox{${#1}$} \copy0 \kern-\wd0
      \kern.2pt \copy0 \kern-\wd0 \kern.2pt
    \copy0 \kern-\wd0 \kern.2pt \box0}
\def\pmxb#1{\setbox0\hbox{${#1}$} \copy0 \kern-\wd0
      \kern.2pt \copy0 \kern-\wd0 \kern.2pt
      \copy0 \kern-\wd0 \kern.2pt \copy0 \kern-\wd0 \kern.2pt \box0}
\def\pmxbb#1{\setbox0\hbox{${#1}$} \copy0 \kern-\wd0 \kern.2pt
      \copy0 \kern-\wd0 \kern.2pt
      \copy0 \kern-\wd0 \kern.2pt \copy0 \kern-\wd0 \kern.2pt
      \copy0 \kern-\wd0 \kern.2pt \box0}

\mathchardef\za="710B  
\mathchardef\zb="710C  
\mathchardef\zg="710D  
\mathchardef\zd="710E  
\mathchardef\zve="710F 
\mathchardef\zz="7110  
\mathchardef\zh="7111  
\mathchardef\zvy="7112 
\mathchardef\zi="7113  
\mathchardef\zk="7114  
\mathchardef\zl="7115  
\mathchardef\zm="7116  
\mathchardef\zn="7117  
\mathchardef\zx="7118  
\mathchardef\zp="7119  
\mathchardef\zr="711A  
\mathchardef\zs="711B  
\mathchardef\zt="711C  
\mathchardef\zu="711D  
\mathchardef\zvf="711E 
\mathchardef\zq="711F  
\mathchardef\zc="7120  
\mathchardef\zw="7121  
\mathchardef\ze="7122  
\mathchardef\zy="7123  
\mathchardef\zf="7124  
\mathchardef\zvr="7125 
\mathchardef\zvs="7126 
\mathchardef\zf="7127  
\mathchardef\zG="7000  
\mathchardef\zD="7001  
\mathchardef\zY="7002  
\mathchardef\zL="7003  
\mathchardef\zX="7004  
\mathchardef\zP="7005  
\mathchardef\zS="7006  
\mathchardef\zU="7007  
\mathchardef\zF="7008  
\mathchardef\zW="700A  

\newcommand{\be}{\begin{equation}}
\newcommand{\ee}{\end{equation}}
\newcommand{\ra}{\rightarrow}

\newcommand{\bea}{\begin{eqnarray}}
\newcommand{\eea}{\end{eqnarray}}
\newcommand{\beas}{\begin{eqnarray*}}
\newcommand{\eeas}{\end{eqnarray*}}
\def\*{{\textstyle *}}
\newcommand{\R}{{\mathbb R}}

\newcommand{\C}{{\mathbb C}}

\newcommand{\D}{{\rm d}}

\newcommand{\we}{\wedge}

\newcommand{\ti}{\times}
\newcommand{\A}{{\cal A}}

\newcommand{\h}{{\sf H}}

\def\ran{\rangle}

\def\op{\oplus}

\def\cD{{\cal D}}

\def\cL{{\cal L}}

\def\cR{{\cal R}}

\def\cH{{\cal H}}
\def\cO{{\cal O}}

\def\wt{\widetilde}
\def\ol{\overline}

\def\Z{\mathbf{Z}}

\def\la{\langle}
\def\ran{\rangle}

\def\br{{\bi r}}

\def\sD{{\sss D}}

\def\sT{{\sss T}}

\def\st{{\sss t}}

\def\xd{\tx{d}}
\def\xi{\tx{i}}

\def\dt{\xd_{\sss T}}

\def\cD{{\cal D}}

\newdir{|>}{%
!/4.5pt/@{|}*:(1,-.2)@^{>}*:(1,+.2)@_{>}}

\def\pr{\operatorname{pr}}

\def\gl{\operatorname{gl}}
\def\GL{\operatorname{GL}}
\def\L{\operatorname{L}}
\def\U{\operatorname{U}}
\def\Id{\operatorname{Id}}
\def\PU{\operatorname{PU}}

\def\u{\operatorname{u}}
\def\h{\operatorname{h}}
\def\P{\mathbf{P}}
\def\D{\mathbf{D}}

\newdir{ (}{{}*!/-5pt/@^{(}}

\newcommand{\tr}{\mbox{$\mathrm{tr}$}}
\newcommand{\bra}[1]{\ensuremath{\langle #1 |}}
\newcommand{\ket}[1]{\ensuremath{| #1\rangle}}
\newcommand{\bk}[2]{\ensuremath{\langle #1 | #2\rangle}}
\newcommand{\kb}[2]{\ensuremath{| #1\rangle\!\langle #2 |}}
\newcommand{\brk}[2]{\ensuremath{( #1 | #2)}}
\newcommand{\gh}{{g_\cH}}
\newcommand{\oh}{{\zw_\cH}}
\newcommand{\oP}{{\zw_{\P(\cH)}}}
\newcommand{\im}{{\mathcal{I}m}}
\newcommand{\re}{{\mathcal{R}e}}

\newcommand{\nm}[1]{\ensuremath{\Vert #1 \Vert}}
\newcommand{\norm}[1]{\left\lVert#1\right\rVert}
\newtheorem{theorem}{Theorem}[section]
\newtheorem{proposition}{Proposition}[section]
\newtheorem{lemma}{Lemma}[section]
\newtheorem{corollary}{Corollary}[section]

\theoremstyle{definition}
\newtheorem{example}{Example}[section]
\newtheorem{definition}{Definition}[section]
\newtheorem{remark}{Remark}[section]

\newcommand{\epf}{\hfill$\square$}

\begin{document}

\title{GEOMETRY OF QUANTUM DYNAMICS \\ IN INFINITE DIMENSION}
\author{Janusz Grabowski\footnote{email: jagrab@impan.pl} \\
\textit{Institute of Mathematics, Polish Academy of Sciences} \\
\\
Marek Ku\'s\footnote{email: marek.kus@cft.edu.pl}\\
\textit{Center for Theoretical Physics, Polish Academy of Sciences}
 \\
\\
Giuseppe Marmo\footnote{email: marmo@na.infn.it}\\
\textit{Dipartimento di Fisica, Universit\`{a} ``Federico II'' di Napoli} \\
\textit{and Istituto Nazionale di Fisica Nucleare, Sezione di Napoli} \\
\\ Tatiana Shulman\footnote{email: tshulman@impan.pl}\\
\textit{Institute of Mathematics, Polish Academy of Sciences} \\
}
\date{}
\maketitle

\maketitle

\begin{abstract}
We develop a geometric approach to quantum mechanics based on the concept of
the Tulczyjew triple. Our approach is genuinely infinite-dimensional and including a Lagrangian formalism in which self-adjoint (Schr\"odinger) operators are obtained as Lagrangian submanifolds associated with the Lagrangian. As a byproduct we obtain also results concerning coadjoint orbits of the unitary group in infinite dimension, embedding of the Hilbert projective space of pure states in the unitary group, and an approach to self-adjoint extensions of symmetric relations.

\end{abstract}

\section{Introduction}
There is a widespread belief among physicists that the Lagrangian and Hamiltonian descriptions
of evolutionary systems are essentially equivalent. That this is not true is quite clear from the considerations of fundamental interactions described by gauge theories. Indeed, to deal with these situations Dirac and Bergmann (separately) elaborated the so called theory of constraints for degenerate Lagrangians \cite{anderson51,bergmann55,dirac50,dirac58} (see also \cite{marmo83}).

A clear cut geometrical approach to these cases was proposed by Tulczyjew \cite{tul1} (see also \cite{TU}). What emerges from this approach may be summarized in a sentence: Lagrangian description provides an implicit differential equation, while the Hamiltonian one provides an explicit differential equation.

In quantum mechanics, due to the probabilistic interpretation, one usually assumes that the evolution is described by a one-parameter group of unitary transformations. By means of the Stone theorem, this assumption requires the infinitesimal generator to be an essentially self-adjoint operator acting on the separable complex Hilbert space $\cH$ associated with the physical system. As it is well known, for instance to properly deal with the canonical commutation relations, one is obliged to consider unbounded operator as the one describing the evolution.
These operators give rise to what is known as the ``domain problem''. It means that in general the Schr\"odinger-type equation one writes to describe the evolution is not defined on the total Hilbert space but only on a subset and one is forced to study the problem of extending this subset to turn the operator into an essential self-adjoint operator, so that it may be integrated to a one-parameter group of unitary transformations.

In recent times (\cite{ashtekar96,schilling96} and \cite[Chapter 6]{carinena15}) it has been observed that within a geometric approach to quantum mechanics, symplectic, Riemannian and complex structures naturally emerge from the Hermitian structure one uses to define the inner product in the Hilbert space. Thus, these geometrical structures naturally call for a more geometrical description of the problem of motion, i.e., the study of self-adjoint operators in geometrical terms.

The aim of this paper is to provide such a geometrical description for the study of self-adjoint operators and their ``domains of self-adjointness".
In particular, we view such domains as the constraints in the Lagrangian picture in which the self-adjoint operators are identified with certain Lagrangian submanifolds of the `symplectic manifold' $\cH\oplus\cH$. Consequently, the Cayley transform is viewed as a symplectomorphism between two such symplectic structures which clearly maps Lagrangian submanifolds onto Lagrangian submanifolds and immediately leads to the von Neumann theorem describing self-adjoint extension of symmetric operators.

In our approach we follow the Tulczyjew's framework for Lagrangian and Hamiltonian formalism viewing consequently the quantum dynamics (Schr\"odinger operators) as certain Lagrangian submanifolds.

Contrary to many works on the geometrical quantum mechanics, we do not restrict considerations to the Hamiltonian picture and do not reduce them to finite-dimensional Hilbert spaces, carefully describing the topological and geometrical structures in the infinite dimension, in particular the Hilbert projective space.

As a byproduct, studying coadjoint orbits of the unitary group of the Hilbert space we recognize closed orbits as those of finite rank operators, that corrects wrong statements known from the literature. To show that the topology of the Hilbert projective space of pure states coincides with the quotient topology of the orbit space, we find a nice (local) embedding of pure states into the unitary group.

We provide examples of quadratic quantum Lagrangians, also constrained ones, and the corresponding Hamiltonians, as well as some remarks concerning the Heisenberg picture and composite systems.

The paper is organized in the following manner.

In section 2 we present the concept of implicit dynamics in the classical picture, and in section 3 we provide the Tulczyjew's approach to the Lagrangian and Hamiltonian mechanics which we will follow in the quantum case.
Section 4 is devoted to reviewing basic functional analysis needed in quantum mechanics. In section 5 we study coadjoint orbits of the unitary group and section 6 is devoted to the geometry and topology of quantum states. Quantum dynamics in the Tulczyjew picture we present in section 7 with many examples, while in section 8 we study self-adjoint extensions of symmetric relations and present a version of the von Neumann theorem. In the next section we place a few comments on the Heisenberg picture and composite systems and we end up with concluding remarks.

\section{Implicit dynamics in Classical Mechanics}
\subsection{Implicit differential equations}
Let us start with an explanation of what we will understand as implicit dynamics on a manifold $N$.
\begin{definition} An {\it (implicit) first-order ordinary differential equation (implicit dynamics)} on a manifold $N$
will be understood as a subset $\sD$ of the tangent bundle $\sT N$. We say that a smooth curve $\zg:\R\ra N$
(or a smooth path $\zg:[t_0,t_1]\ra N$) {\it satisfies the equation} $\sD$ (or {\it is a solution of $\sD$}),
if its tangent prolongation $\st\zg=(\zg,\dot\zg):\R\ra\sT N$ (resp., $\st{\zg}:[t_0,t_1]\ra\sT N$) takes values in
$\sD$. A curve (or a path) $\wt{\zg}$ in $\sT N$ we call {\it admissible}, if it is the tangent prolongation
of its projection ${\wt{\zg}}_N$ on $N$.
\end{definition}
\begin{example} (explicit differential equations)
Of course, the most studied are first-order \emph{explicit} ordinary differential equation which are given as the range $\sD=X(N)\subset\sT N$ of a vector field $X:N\to\sT N$ on $N$. Solutions in this case are called \emph{trajectories} of $X$.
\end{example}

According to the above definition, solutions of an implicit dynamics $\sD\subset\sT N$ on a manifold $N$ are projections
$\wt{\zg}_N$ of admissible curves $\wt{\zg}$ lying in $D$. Note, however, that different implicit differential
equations may have the same set of solutions. First of all, if $\sD$ is supported on a subset $N_0$,
$\zt_N(\sD)=N_0$, only vectors from $\sD\cap\sT N_0$ do matter if solutions are concerned. Hence,
the \emph{first integrability extract}
\be\label{ie}\sD^1=\sD\cap\sT N_0
\ee
has the same solutions as $\sD$, and $\sD\subset \sT N_0$ is the {\it first
integrability condition}. Of course, replacing $\sD$ with $\sD^1$ may turn out to be an infinite procedure, but this will not happen in examples considered
in this paper. Of course, explicit differential equations are automatically integrable.

\begin{example}
Consider an implicit dynamics on $N=S^1\ti\R^4$ given by $\sD\subset\sT N$ parameterized by $(\zf,y^1,y^2)\in S^1\ti\R^2$ as follows:
\beas \sD&=&\{(\zf,\zx_1,\zx_2,\zx_3,\zx_4,\dot{\zf},\dot{\zx}_1,\dot{\zx}_2,\dot{\zx}_3,\dot{\zx}_4):
\zx_1=J_1y^1\,,\ \zx_2=(mR^2+J_2)y^2\,,\\
&& \zx_3=mRy^2\cos{\zf}\,,\ \zx_4=mRy^2\sin{\zf}\,,\ \dot{\zf}=y^1\,,\ \dot{\zx}_1=0\,,\ \dot{\zx}_2=0\}\,,
\eeas
where $m,R,J_1,J_2$ are some constants. The dynamics is clearly not explicit due to the constraints \be\label{e61}\zx_3=\zm\cos{\zf}\cdot\zx_2,\quad \zx_4=\zm\sin{\zf}\cdot\zx_2,\,
\ee
where $\zm=\frac{mR}{mR^2+J_2}$.
We have the equations
$$\dot{\zf}=\zx_1/J_1\,,\ \dot{\zx}_1=0\,,\ \dot{\zx}_2=0\,,$$
but $\dot{\zx}_3$ and $\dot{\zx}_4$ are arbitrary. It is interesting that the first integrability condition allows us to describe them
as well, since it gives
$$\dot{\zx}_3=-\zm\zx_2\sin{\zf}\cdot\dot{\zf}=-\frac{\zm}{J_1}\zx_1\zx_2\sin{\zf}\,,\quad
\dot{\zx}_4=\zm\zx_2\cos{\zf}\cdot\dot{\zf}=\frac{\zm}{J_1}\zx_1\zx_2\cos{\zf}\,.
$$
Thus, the dynamics is integrable. Actually, it is the properly understood phase dynamics  of vertical rolling disc on a plane in a Dirac algebroid setting studied in \cite{Grabowska:2011}.
\end{example}

All this can be generalized to ordinary implicit differential equations of arbitrary order. In this case we
consider $\sD$ as a subset of higher jet bundles, the $n$-th tangent bundle $\sT^n N$ in case of an equation
of order $n$, and consider $\zg$ as a solution if its $n$-th jet prolongation takes values in $\sD$. If we
call the $n$-th jet prolongations {\it admissible} in $\sT^n N$, then solutions of $\sD$ are exactly
projections $\wt{\zg}_N$ to $N$ of admissible curves (or paths) $\wt{\zg}$ in $\sT^n N$ lying in $\sD$.

\begin{remark} The implicit differential equations described above are called by some authors {\it differential relations}.
Let us explain that we use the most general definition, not requiring from $\sD$ any differentiability
properties, since in real life the dynamics $\sD$ we encounter are often not submanifolds. This generality is
also very convenient, as allows us to skip technical difficulties in the corresponding Lagrangian and
Hamiltonian formalisms. Of course, what is a ballast in defining implicit dynamics can be very useful in
solving the equations, but in our opinion, solving could be considered case by case, while geometric
formalisms of generating dynamics should be as general as possible. Note also that for any subset $N_0$ of a
manifold $N$ the tangent prolongations $\sT N_0$, $\sT^2 N_0$, etc., make precisely sense as subsets of $\sT
N$, $\sT^2 N$, etc. They are simply understood as families of the corresponding jets of appropriately smooth
curves in $N$ which take values in $N_0$.
\end{remark}

\subsection{Lagrangian submanifolds}
Note, however, that for many instances in Classical Mechanics, $N$ is a phase space equipped with a symplectic form $\zw$ and the phase-space implicit dynamics we encounter are not only submanifolds of $\sT N$, but \emph{Lagrangian submanifolds} of $\sT N$.

Let us explain this statement briefly. First, recall that a Lagrangian submanifold $\cL$ of a symplectic manifold $(N,\zw)$ of dimension $2n$ is a submanifold of dimension $n$ on which the symplectic form vanishes. In the case of a standard phase space, $N=\sT^*Q$, equipped with the canonical symplectic form $\zw_Q$, the following is well known.

\begin{proposition}\label{p1} The range $\cL=\za(Q)$ of a one-form $\za$ on $Q$ viewed as a section $\za:Q\to\sT^*Q$ of the cotangent bundle, is a Lagrangian submanifold in $\sT^*Q$ if and only if
$\za$ is a closed form. Moreover, the above Lagrangian submanifolds can be characterized as those
for which the canonical projection $\zp_Q:\sT^*Q\to Q$ induces a diffeomorphism on $\cL$ onto $Q$.
\end{proposition}
In particular,
any function $f:Q\to\R$ generates a Lagrangian submanifold $\xd f(Q)\subset\sT^*Q$ being the image of the differential 1-form $\xd f$.

Let $\cL\subset\sT^*Q$ be a Lagrangian submanifold of a cotangent bundle. A \emph{regular
point} $p\in\cL$ is a point where the derivative of
${\zp_Q}_{|\cL}$ is surjective, i.e.\   $\cL$ is transversal to the fibers of the projection $\zp_Q$.
A nonregular point is called \emph{singular} or \emph{critical}. A singular point of a Lagrangian
submanifold is also called \emph{Lagrangian singularity} or \emph{catastrophe}. The
set of all the points of Q on which are based the singular points of $\cL$ is
called the \emph{caustic} of $\cL$.

The above propositions implies that $\cL$ is the image of a closed 1-form only if it is \emph{regular}, i.e.\   has no singular points.

\subsection{Tangent lifts of forms and Hamiltonian vector fields}
There exists a derivation $\dt$ (cf.\    \cite{GU,YI}) on the exterior algebra of forms on a
manifold $N$ with values in the exterior algebra of forms of the
tangent bundle $\sT N$  which plays essential r\^ole in the calculus of
variations (\cite{tul1,tul}) and in analytical
mechanics. If, in local coordinates, $\zm = \zm_{i_1\dots i_r}\xd x^{i_1}\wedge
\cdots\wedge \xd x^{i_r}$, then
		$$ \dt\zm (x,\dot x) = \frac{\partial \zm_{i_1\dots i_r}}{\partial x^k}(x)
  \dot x^k \xd x^{i_1}\wedge \dots \wedge \xd x^{i_r} + \sum_m
\zm_{i_1\dots i_r}(x) \xd x^{i_1}\wedge \dots\wedge \xd \dot x^{i_m}\wedge \dots \xd x^{i_r}
														$$
for $r>0$ and
		$$ \dt\zm(x, \dot x) = \frac{\partial \zm}{\partial x^i}(x)\dot x^i
																$$
for $r=0$. The form $\dt\zm$ we call the \emph{tangent lift} of $\zm$.

Let $\zm$ be a 2-form on $N$ and $\zm^\flat:\sT N\to \sT^* N$ be the vector bundle morphism induced by the contraction with $\zm$, i.e.\   $\zm^\flat(X_p)=i_{X_p}\zm(p)$. We have the following (\cite{GU,tul1,tul}).
\begin{proposition}\label{p2}
A 2-form $\zm$ is closed if and only if $\dt\zm=(\zm^\flat)^*(\zw_N)$. In particular, if $\zw$ is symplectic, then $\dt\zw$ is symplectic and $\zw^\flat$ is a symplectomorphism between $(\sT N,\dt\zw)$ and $(\sT^*N,\zw_N)$.
\end{proposition}
On a symplectic manifold $(N,\zw)$, locally Hamiltonian vector fields correspond, \emph{via} the isomorphism of vector bundles $\zm^\flat:\sT N\to \sT^* N$, to closed one-forms. Since in the case when the first cohomology group
$H^1(N)$ vanishes, e.g. $N$ is contractible, there is no difference between closed and exact 1-forms, so locally and globally Hamiltonian vector fields, in view of propositions \ref{p1} and \ref{p2}, we get the following.
\begin{proposition}\label{p3}
A vector field $X\colon N\rightarrow \sT N$ on a contractible symplectic manifold $N$ is  Hamiltonian if and only if its image $X(M)$
is a Lagrangian submanifold of $(\sT M, \dt \zw )$.
\end{proposition}
The concept of a generalized Hamiltonian system can be introduced as a Lagrangian
submanifold of $(\sT N, \dt \zw)$. The infinitesimal  dynamics
of a relativistic particle is an example of such a system.

\begin{example}\label{rp} (cf.\    \cite{TU})
The (implicit) phase-space dynamics of a free relativistic particle in a Minkowski space $Q$ is  described by equations
        \bea\label{rp1}
            0 &= g^{\zk\zl}p_\zk p_\zl \\
            \dot q^\zk &= v g^{\zk\zl}p_\zl \\
            \dot p_\zk &= -\frac{v}{2} \partial_\zk g^{\zm\zn}p_\zm p_\zn\,,
        \eea
where $g_{\zk\zl}$ is the Minkowski metric and $v>0$ (more precisely, we should take only the `future part' of (\ref{rp1})). The equations describe a Lagrangian submanifold $\sD$ in $\sT\sT^*Q$ which is not the range of any vector field on $\sT^*Q$ due to the constraint
$g^{\zk\zl}p_\zk p_\zl =0 $. However, following Tulczyjew \cite{tul1,tul}, it is possible to obtain the above dynamics from a \emph{constrained Lagrangian}, as we explain in the next section.
\end{example}

\section{The Tulczyjew triple}
The canonical symplectic form $\zw_M$ on $\sT^\* M$ induces an isomorphism
$$\zb_M=(\zw_M)^\flat:\sT\sT^\* M\to\sT^\*\sT^\* M\,.$$
Composing it with $\cR_{\sT M}$, where, for any vector bundle $E$,
$$\cR_E:\sT^\*E^\*\to \sT^\*E$$
is the well-known canonical isomorphism (see e.g. \cite{KU,Ur}),
 we get the map
$$\za_M:\sT\sT^\* M\to\sT^\*\sT M\,.$$
Using the standard coordinates $(x^\zm,\dot{x}^\zn)$ and $(x^\zm,p_\zn)$ on $\sT M$ and $\sT^*M$, respectively, and the adapted coordinates on $\sT^*\sT M$ and $\sT\sT^*M$, we can write
\be\label{alpha} \za(x,p,\dot x,\dot p)=(x,\dot x,\dot p,p)\,.
\ee
This gives rise to the commutative diagram of \emph{double vector bundle isomorphisms} and symplectomorphisms (the \emph{Tulczyjew triple})
\be\label{tt}
    \xymatrix@R-4mm @C-10mm
        { & \sT^\*\sT^\*M \ar[ldd]_*{} \ar[rd]^*{} & & & \sT \sT^\* M \ar[rrr]^*{{\za}_M}
        \ar[lll]_*{\zb_M} \ar[ldd]^*{} \ar[rd]^*{}& & & \sT^\*\sT M \ar[ldd]^*{} \ar[rd]^*{} & \cr
        & & \sT M \ar[ldd]^*{} & & & \sT M  \ar[ldd]^*{} \ar@{=}[lll]^*{} \ar@{=}[rrr]^*{} & & & \sT M \ar[ldd]^*{}  \cr
         \sT^\* M  \ar[rd]^*{}  & & & \sT^\* M \ar@{=}[rrr]^*{} \ar@{=}[lll]^*{} \ar[rd]^*{}& & & \sT^\* M
        \ar[rd]^*{} & & \cr
        & M  & & &  M \ar@{=}[rrr]^*{} \ar@{=}[lll]^*{} & & & M & }.
\ee

The map $\za_M$  (or $\zb_M$) encodes the \emph{Lie algebroid structure} of $\sT M$, i.e.\   the Lie bracket of vector fields (cf.\    \cite{GU3,GU1}).

The Lagrangian and Hamiltonian formalisms have simple description in terms of the Tulczyjew triple. The true physical dynamics, the \emph{phase dynamics}, will be described as an implicit first order differential equation on the \emph{phase space} $\sT^*M$, given by a submanifold $\sD\subset\sT\sT^*M$. Note that this picture, together with the implication to geometrical mechanics, can be easily extended to the case of an arbitrary Lie algebroid (or even a \emph{general algebroid} in the sense of \cite{GU3,GU1}), as shown in \cite{Grabowska:2008,Grabowska:2006}.

\subsection{The Tulczyjew triple - Lagrangian formalism}

Starting with a Lagrangian $L:\sT M\rightarrow \R$ we derive the diagram

$$\xymatrix@C-20pt@R-10pt{
{\sD}\ar@{ (->}[r]& \sT\sT^\ast M \ar[rrr]^{\alpha_M} \ar[dr]\ar[ddl]
 & & & \sT^\ast\sT M\ar[dr]\ar[ddl] & \\
 & & \sT M\ar@{.}[rrr]\ar@{.}[ddl]
 & & & \sT M \ar@{.}[ddl]\ar@/_1pc/[ul]_{\xd L}\ar[dll]_{\zl_L}\ar[ullll]_{\mathcal{T}L}\\
 \sT^\ast M\ar@{.}[rrr]\ar@{.}[dr]
 & & & \sT^\ast M\ar@{.}[dr] & &  \\
 & M\ar@{.}[rrr]& & & M &
}$$
The dynamics
$$\sD=\alpha_M^{-1}(\xd L(\sT M)))=\mathcal{T}L(\sT M)\,,$$
is given as the range of the \emph{Tulczyjew differential} $\mathcal{T}L=\za_M^{-1}\circ\xd L$.
In local coordinates,
\be\label{phd}\sD=\left\{(x,p,\dot x,\dot p):\;\; p=\frac{\partial L}{\partial \dot x},\quad \dot p=\frac{\partial L}{\partial x}\right\}\,.
\ee
In the diagram we indicated also the \emph{Legendre map},
\be\label{leg}\zl_L:\sT M\rightarrow \sT^\ast M, \;\; \quad \zl_L(x,\dot x)=
(x,\frac{\partial L}{\partial \dot x})\,.
\ee

\subsection{Euler-Lagrange equations}
Let now, $\zg:\R\to M$ be a curve in $M$ (of course, $\R$ can be replaced by an open interval), and $\st\zg:\R\to\sT M$ be its tangent prolongation. It is easy to see that both curves, $\xd L\circ\st\zg$ and $\za_M\circ\st(\zl_L\circ\st\zg)$ are curves in $\sT^*\sT M$ covering $\st\zg$. Therefore, their difference makes sense and, as easily seen, takes values in the annihilator $V^0\sT M$ of the vertical subbundle $V\sT M\subset\sT\sT M$. Since $V^0\sT M\simeq\sT M\ti_M\sT^*M$, we obtain a map
$\zd L_\zg:\R\to\sT^*M$. The above map is interpreted as the external force along the trajectory. Its value at $t\in\R$ depends on the second jet $\st^2\zg(t)$ of $\gamma$ only, so defines the variation of the Lagrangian, understood as a map
\be\label{work}\zd L:\sT^2M\to\sT^*M\,,
\ee
where $\sT^2M$, the second tangent bundle, is the bundle of all second jets of curves $\R\to M$ at $0\in\R$.
The equation
\be\label{EL}\zd L_\zg=\zd L\circ\st^2\zg=0
\ee
is known as the \emph{Euler-Lagrange equation} and tells that the curve $\xd L\circ\st\zg$ corresponds \emph{via} $\za_M$ to an \emph{admissible curve} in $\sT\sT^*M$, i.e.\   the tangent prolongation of a curve in $\sT^*M$. Here, of course, $\st^2\zg$ is the second tangent prolongation of $\zg$ to $\sT^2M$.

From (\ref{phd}) we get immediately the Euler-Lagrange equations in the form
$$\frac{\partial L}{\partial x}=\frac{\xd}{\xd t}\left(\frac{\partial L}{\partial \dot{x}}\right)\,.$$

\subsection{The Tulczyjew triple - Hamiltonian formalism}
The Hamiltonian formalism looks analogously. If {$H:\sT^\ast M\rightarrow \R$} is a Hamiltonian function, from the Hamiltonian side of the triple
{$$\hskip-1.2cm
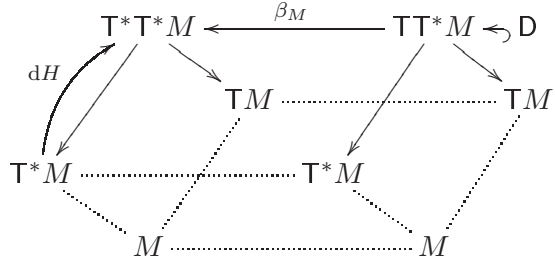
\xymatrix@C-20pt@R-10pt{
 & \sT^\ast\sT^\ast M  \ar[dr] \ar[ddl]
 & & & \sT\sT^\ast M\ar[dr]\ar[ddl] \ar[lll]_{\beta_M}&
 { \sD}\ar@{ (->}[l] \\
 & & \sT M\ar@{.}[rrr]\ar@{.}[ddl]
 & & & \sT M \ar@{.}[ddl]\\
 \sT^\ast M\ar@{.}[rrr]\ar@{.}[dr] \ar@/^1pc/[uur]^{\xd H}
 & & & \sT^\ast M\ar@{.}[dr] & &  \\
 & M\ar@{.}[rrr]& & & M &
}\qquad$$}
we derive the phase-space dynamics in the form
{$$\sD=\beta_M^{-1}(\xd H(\sT^\ast M))\,.$$}
It is automatically explicit, i.e.\   generated by the corresponding Hamiltonian vector field, so corresponds to a phase dynamics induced by a Lagrangian function only in regular cases.
In local coordinates,
{$$\sD=\left\{(x,p,\dot x,\dot p):\;\; \dot p=-\frac{\partial H}{\partial x},\quad \dot x=\frac{\partial H}{\partial p}\right\}\,,$$}
so we obtain the standard Hamilton equations.

The question of finding Hamiltonian description of the phase dynamics associated with a Lagrangian $L:\sT M\to\R$ can be now easily explained as the question of finding a Hamiltonian $H:\sT^*M\to\R$ such that $\zb^{-1}(\xd H(\sT^*M))=\za^{-1}(\xd L(\sT M))$, i.e.\   such that the Lagrangian submanifold $\xd H(\sT^*M)$ corresponds to the Lagrangian submanifold $\xd L(\sT M)$ \emph{via} the anti-symplectomorphism $\cR_{\sT M}$. It is always possible if the Lagrangian is hyper-regular.
If, however, we have started from a Lagrangian which is not regular, or we assume some constraints, the resulted phase dynamics may fail to come from a Hamiltonian, so $\sD$ being still Lagrangian submanifold is not of the form $X(N)$. A partial solution is to add also constrained Hamiltonians
into the picture.

\subsection{Constrained Lagrangians and Hamiltonians}
Starting with a \emph{constrained Lagrangian}, i.e.\   a Lagrangian defined only on a constraint manifold $S\subset {\sT M}$, we can slightly modify the above picture and get the diagram
$$
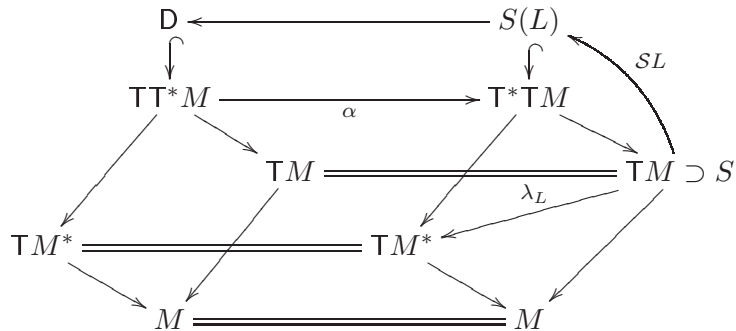
\xymatrix@C-10pt@R-10pt{
 & & {\sD}\ar@{ (->}[d] & & & {S(L)}\ar[lll]\ar@{ (->}[d] & \\
  & & \sT {\sT^\ast M}  \ar[dr] \ar[ddl]\ar[rrr]_{\za}
 & & &{ \sT^\ast {\sT M}}\ar[dr]\ar[ddl] & \\
 & & & \sT M\ar[ddl]
 & & & {{\sT M}\supset {S}} \ar[ddl]\ar@{=}[lll]\ar@/_1pc/[uul]_{\mathcal{S} L}\ar[dll]_{\lambda_L}  \\
 &  {\sT M}^\ast\ar[dr]
 & & & {\sT M}^\ast\ar@{=}[lll]\ar[dr] & &  \\
 & & M  & & & M \ar@{=}[lll]&
}$$
Here, $S(L)$ is the Lagrangian submanifold in $\sT^\ast {\sT M}$ induced by the constrained Lagrangian and ${\mathcal{S} L}:S\to\sT^*{\sT M}$  is the corresponding relation. Recall that, in general, any smooth function $L:S\to\R$ defined on a submanifold $S$ of $N$ generates the Lagrangian submanifold $S(L)$ of $\sT^*N$:
\be\label{cl}S(L)=\{ \za_e\in\sT^\ast_e{N}: e\in S\text{\ and\ }\la\za_e,v_e\ran=\xd L(v_e)\text{\ for every\ } v_e\in\sT_eS\}\,.
\ee
Lagrangian submanifolds of this kind have been introduced in \cite{tul} and in the case $S=N$ we get the well-known Lagrangian submanifolds of $\sT^*N$ of the form $\xd L(N)$.
The \emph{vakonomically constrained phase dynamics} is just $\sD=\za^{-1}(S(L))\subset\sT {\sT^\ast M}$.  We stress that, due to the fact that we are dealing with a vakonomically constrained system, relations and not just genuine smooth maps naturally appear in the formalism.

Completely analogously we can obtain constrained phase dynamics $\sD\subset \sT N$ from a constrained Hamiltonian $H:N\supset S\to\R$ defined on a symplectic manifold $(N,\zw)$, namely
\be\label{cH}
\sD=(\zw^\flat)^{-1}(S_H)\,.
\ee
\begin{example} The implicit dynamics $\sD$ of a free relativistic particle described in example \ref{rp} is of the form (\ref{cH}), with the trivial Hamiltonian $H=0$ defined on the constraint
$S\subset\sT^*Q$ being the `future part' of the cones $g^{\zk\zl}p_\zk p_\zl =0 $.
\end{example}

Our aim is to develop an analogous geometric and rigorous approach to Quantum Mechanics in the spirit of \cite{ashtekar96,grabowski05} but based on Tulczyjew triples (so including the Lagrangian part of the theory) and in the most physically interesting case of an infinite-dimensional Hilbert space. Note that this approach allows in principle also for non-linear dynamics, e.g. nonlinear Schr\"odinger equations.

\section{The unitary group and topology}
To fix the terminology and notation, let us recall that the main geometrical objects of Quantum Mechanics are associated with a (complex) Hilbert space $\cH$ which is equipped with a Hermitian inner product $\bk{\cdot}{\cdot}$ and the corresponding norm
$$\| x\|=\| x\|_\cH:=\sqrt{\bk{x}{x}}\,.$$
The Hilbert space is generally infinite-dimensional, but we will assume it is separable, i.e.\   $\cH$ admits a countable orthonormal basis.

First of all, one can consider the group of automorphisms of the Hilbert space, preserving the complex linear and Hermitian structures, i.e.\   the unitary group $\U(\cH)$. This group is naturally included in the group $\GL(\cH)$ of invertible elements of the algebra $\gl(\cH)$ of all
continuous complex linear maps $A:\cH\to\cH$. Note that $\gl(\cH)$ is a $C^*$-algebra with respect to the operator norm
\be\label{operatornorm}\| A\|=\sup\{\| Ax\|\,: x\in\cH\,,\quad \| x\|\le 1\}
\ee
and the $*$-operation being the Hermitian conjugation $A\mapsto A^\dag$, where
\be\label{dag}\bk{A^\dag x}{y}=\bk{x}{Ay}\,.\ee
Operators from $\gl(\cH)$ satisfying $A^\dag=A$ we call \emph{Hermitian}; those with $A^\dag=-A$ \emph{anti-Hermitian}. With $\h(\cH)$ we will denote the (real) Banach subspace of Hermitian) operators.
\begin{remark}
Actually, the adjoint operator $A^\dag$ makes sense even for densely defined operators $A:\cH\supset\cD\to\cH$. The domain of $A^\dag$ consists of those $x\in\cH$ for which the functional $$\cD\ni y\mapsto \bk{x}{Ay}\in\C$$ is continuous and (\ref{dag}) defines $A^\dag$ on this domain.
\end{remark}

Since $\GL(\cH)$ is an open subset in $\gl(\cH)$, it is easy to see that the group $\GL(\cH)$ is a (complex) Banach-Lie group modelled on $\gl(\cH)$. What is more, $U(\cH)$ is its (real) Lie subgroup
defined as a level set of the smooth map
$$\gl(\cH)\ni A \mapsto AA^\dag\in h(\cH)\,.$$
Indeed, elements $U$ from $\U(\cH)$ are characterized by $UU^\dag=I$ and the above map is a submersion in a neighbourhood of $\U(\cH)$ (cf.\   \cite{bourbaki1}). Consequently, the Lie algebra of $\U(\cH)$ is the real Banach subspace $\u(\cH)$ of $\gl(\cH)$ consisting of anti-Hermitian  operators and equipped with the commutator bracket.
\begin{remark}
Besides the norm topology, the unitary group carries another topology in which it is also a topological group (cf.\   \cite{schott}). This is the \emph{strong topology}:
$$U_k\rightarrow U\quad\Leftrightarrow\quad \forall\, x\in\cH\,[\,U_kx\rightarrow Ux\,]\,.$$
However, $\U(\cH)$ is not a Lie group with respect to the strong topology. The strong one-parameter subgroups are known to be generated by (generally unbounded and only densely defined) anti self-adjoint operators, i.e.\   operators $iA$, where $A$ is self-adjoint (Stone theorem). However, the set of general anti self-adjoint operators does not carry the structure of a Lie algebra: the commutator and even the addition is not well defined.
The self-adjoint operators, in turn, are usually interpreted as \emph{quantum observables}.
\end{remark}

The operators of the form $AA^\dag$ are called \emph{positive semi-definite}, as they coincide with operators $T\in \gl(\cH)$ such that $\bk{Tx}{x}\ge 0$ for all $x\in\cH$. For $A\in \gl(\cH)$ we have the \emph{polar decomposition} $A=U|A|$, where $|A|=\sqrt{AA^\dag}$ and $U$ is a unitary operator.

It is well known that the $C^*$-algebra $\gl(\cH)$ contains the ideal $\gl_0(\cH)$ of finite-rank operators and \emph{Schatten ideals} $\L^p(\cH)$, $1\le p< \infty$, where $\L^p(\cH)$ consists of those $A\in \gl(\cH)$ for which $\tr|A|^p<\infty$. Note that $\L^p(\cH)$ is a Banach algebra with respect to the norm
\be\label{pnorm}\| A\|_p=\left(\tr|A|^p\right)^{\frac{1}{p}}\,.
\ee
One can easily view each of $\L^p(\cH)$ as the completion of $\gl_0(\cH)$ with respect to the norm (\ref{pnorm}). This allows to define $\L^\infty(\cH)$ as the completion of $\gl_0(\cH)$ with respect to the norm
\be\label{inftynorm}\| A\|_\infty=\lim_{p\to\infty}\left(\tr|A|^p\right)^{\frac{1}{p}}\,.
\ee
The above norm on $\gl_0(\cH)$ coincides with the operator norm (\ref{operatornorm}) and the resulted closed ideal in $\gl(\cH)$ consists of \emph{compact operators}.
With $\h^p(\cH)\subset \L^p(\cH)$ we will denote the corresponding (real) subspaces of Hermitian operators from $\L^p(\cH)$. The unitary group acts on each of these spaces by $A\mapsto UAU^\dag$.

As for the duality we have the following well-known result.
\begin{proposition}
For $1\le p\le p'\le\infty$, we have
\be\label{inclusions} \L^p(\cH)\subset \L^{p'}(\cH)\subset \gl(\cH).
\ee
For $1<p<\infty$, the dual Banach space of $\L^p(\cH)$ is $\L^q(\cH)$, where $\frac{1}{p}+\frac{1}{q}=1$. Moreover,
\be\label{duality}
\left(\L^\infty(\cH)\right)^*=\L^1(\cH)\,,\quad\text{and}\quad \left(\L^1(\cH)\right)^*=\gl(\cH)\,.
\ee
The above dualities come from the pairings
\be\label{pairing}\L^p(\cH)\ti\L^q(\cH)\ni(A,B)\mapsto \langle A,B\rangle^p_q=\tr(AB)\in\C\,,
\ee
where $1<p<\infty$ and $\frac{1}{p}+\frac{1}{q}=1$, or $p=\infty$ and $q=1$, and the pairing
\be\label{pairing0}\L^1(\cH)\ti\gl(\cH)\ni(A,B)\mapsto \langle A,B\rangle=\tr(AB)\in\C\,.
\ee
\end{proposition}
\begin{remark}
Let us note that in the literature the symbol $\L^\infty(\cH)$ refers often to $\gl(\cH)$ and not to the space of compact operators. It seems, however, that our use of the symbol is logically justified.
\end{remark}
According to a general theory \cite{OR,Ratiu}, $\L^1(\cH)$ is canonically a \emph{Banach Lie-Poisson space} as a predual of a $W^*$-algebra, namely of $\gl(\cH)$, and the spaces $\L^p(\cH)$, for $1<p<\infty$, are canonically \emph{Banach Lie-Poisson spaces} by being reflexive. The canonical linear Poisson bracket $\{\cdot,\cdot\}_p$ is given by the Kostant-Kirillov-Souriau formula
\be\label{PB}
\{ f,g\}_p(A)=\langle A,[\xd f(A),\xd g(A)]\rangle^p_q
\ee
for $1<p<\infty$, and
\be\label{PB0}
\{ f,g\}_1(A)=\langle A,[\xd f(A),\xd g(A)]\rangle
\ee
for $p=1$. Here, $\xd f(A),\xd g(A)$ are understood as elements of the dual space and the bracket $[\cdot,\cdot]$ is the commutator.

\medskip
There are two particularly important Schatten ideals: $\L^2(\cH)$, the space of \emph{Hilbert-Schmidt operators}, which is a Hilbert space itself with respect to the Hermitian inner product (we will often skip the subscript ``HS'' for this product)
\be\label{HS}
\bk{A}{B}_{HS}=\tr(A^\dag B)\,,
\ee
and the space $\L^1(\cH)$ of \emph{nuclear} (\emph{trace-class}) \emph{operators} on which we have
a distinguished functional of the trace:
\be\label{tr}
\tr:\L^1(\cH)\to\C\,,\quad A\mapsto \tr(A)\,.
\ee
Note that the Hermitian inner product (\ref{HS}) induces a scalar product $\brk{\cdot}{\cdot}$ on the space $\h^2(\cH)$ of Hermitian Hilbert-Schmidt operators, $\brk{A}{B}=\tr(AB)$, so one can view $\h^2(\cH)$ as a canonical Euclidean space. We will also view $\h^1(\cH)$ as the \emph{predual space} $\u_*(\cH)$ of the Lie algebra $\u(\cH)=i\cdot \h(\cH)$ of the unitary group, with the pairing
\be\label{predual}
\u_*(\cH)\ti\u(\cH)\ni(A,B)\mapsto i\cdot\tr(AB)\in\R\,.
\ee
From this point of view, the natural action of $\U(\cH)$ on $\u_*(\cH)$ can be seen as the co-adjoint action, so the orbits carry symplectic structures. However, the situation is much more complicated than in finite dimensions. The orbits are in general only \emph{weakly immersed} submanifolds and the symplectic structures are \emph{weakly symplectic}. There is an extensive literature on the subject (see e.g. \cite{Ratiu} and references therein), but we will concentrate on physically important examples and will not develop a general theory. This will allow for relatively easy proofs of the facts we will need in the sequel.

\section{Coadjoint orbits in $\u_*(\cH)$}
In \cite{Bona} there were investigated infinite dimensional unitary coadjoint orbits of symmetric trace-class operators.
On each such orbit there are 2 natural topologies - the topology of Banach space of trace-class operators and the topology induced from
the coadjoint action of the unitary group $U(\cH)$. It was shown in \cite{Bona} that for orbits going through finite rank operators, these topologies coincide, or, in other words, that the orbits of finite rank operators are immersed into the space of trace class operators. What is more, it is esy to see (cf.\    \cite{Bona}) that also all $\L^k$-topologies, $k\ge 1$, coincide  on an orbit $\cO_\zr$ of a finite-rank $\zr\in\u_*(\cH)=\h^1(\cH)$.

\subsection{Closedness problem}
In \cite{Bona} it was erroneously claimed that
all unitary coadjoint orbits of Hermitian trace-class operators are closed in the trace-class topology. The mistake came from the wrong statement that the orbit of $\rho$ is completely determined by nonzero eigenvalues of $\rho$ and their multiplicities. Here we prove that the orbit $\cO_\zr  $ is closed if and only if $\rho$ is finite rank.

\begin{theorem} Let $\rho\in \u_*(\mathcal H)$. Then $\cO_\zr  $ is closed in the $\L^1$-topology if and only if $\rho$ is of finite rank.
\end{theorem}
\begin{proof} "{\bf Only if}":  Suppose $\rho$ has infinite rank.  We will consider 3 different cases: when $0$ is not an eigenvalue of $\rho$, when it is an eigenvalue of finite multiplicity and when it is an eigenvalue of infinite multiplicity.

{\bf Case 1}: 0 is not an eigenvalue of $\rho$. Then in some orthonormal basis $\rho$ has diagonal form $$\rho = diag(a_1, a_2, \ldots)$$ with all $a_i> 0$ and such  that $\sum_{k=1}^{\infty} a_k < \infty$. Let $$\rho_n = diag(a_{n+1}, a_1, a_2, \ldots, a_n, a_{n+2}, a_{n+3}, \ldots)$$ with respect to the same basis.  Then $\rho_n \in \cO_\zr  $. Let $$ \rho' = diag(0, a_1, a_2, \ldots)$$ with respect to the same basis. Then
\begin{multline}\|\rho_n -\rho'\|_1 = \|diag(a_{n+1}, 0, \ldots, 0, a_{n+2}-a_{n+1}, a_{n+3}-a_{n+2}, \ldots)\|_1 = \\ a_{n+1} + \sum_{k=1}^{\infty} |a_{n+k+1}-a_{n+k}| \le 2 \sum_{k=n+1}^{\infty} a_k \to 0  \end{multline}
as $n\to \infty$. Hence $\rho' \in \overline{\cO_\zr  }$.  However $\rho'\notin \cO_\zr  $ because $\rho'$ has an eigenvalue 0. Thus $\cO_\zr  $ is not closed.

{\bf Case 2}: 0 is an eigenvalue of multiplicity $m< \infty$. Then in some orthonormal basis $\rho$ has diagonal form $$\rho = diag(0, \ldots, 0, a_1, a_2, \ldots)$$ (here 0 is repeated m times)  with all $a_i> 0$ and such  that $\sum_{k=1}^{\infty} a_k < \infty$. Let $$\rho_n = diag(a_1, \ldots, a_n, 0,\ldots, 0, a_{n+1}, a_{n+2}, \ldots),$$ with respect to the same basis and where again 0 is repeated m times. Then $\rho_n \in \cO_\zr  $. Let $$ \rho' = diag(a_1, a_2, \ldots)$$ with respect to the same basis. Then \begin{multline}\|\rho_n - \rho'\|_1 = \|diag(0, \ldots, 0, -a_{n+1}, \ldots, -a_{n+m}, a_{n+1} - a_{n+m+1}, a_{n+2} - a_{n+m+2}, \ldots)\| = \\ \sum_{k=1}^m a_{n+k} + \sum_{k=n+1}^{\infty} |a_k - a_{k+m}| \le 3 \sum_{k=n+1}^{\infty} a_k \to 0
\end{multline} as $n\to \infty$. Thus $\rho' \in \overline{\cO_\zr  }$ but $\rho' \notin \cO_\zr  $, since $\rho'$ does not have an eigenvalue 0.

{\bf Case 3}: 0 is an eigenvalue of infinite multiplicity. Let $\mathcal H_1 = Ker \rho$, $\mathcal H_2 = \mathcal H_1^{\perp}$, $f_1, f_2, \ldots$ an orthonormal basis in $\mathcal H_1$, $ e_1, e_2, \ldots$ an orthonormal basis in $\mathcal H_2$.  Then $\rho f_i = 0$ for all i's and $\rho e_i = a_i e_i$ for some positive $a_i$ such that $\sum_{k=1}^{\infty} a_k < \infty$. Define  $\rho_n$ by $$\rho_n f_i = 0, i\ge n,$$ $$ \rho_n f_i = a_{2i}f_i, i\le n-1, \;\;$$ $$\rho_n e_i = a_{2i-1}e_i, i \le n-1, \;\;$$$$ \rho_n e_i = 0, n \le i < 2n, \;\; $$$$\rho_ne_i = a_ie_i, i\ge 2n.$$ Then $\rho_n \in \cO_\zr  $. Define $\rho'$ by $$\rho' f_i = a_{2i}f_i, \; i\in \mathbb N$$  $$\rho'e_i = a_{2i-1}e_i,\; i \in \mathbb N. $$
Then,
$$(\rho_n - \rho')f_i =
(\rho_n - \rho')e_i = 0, \; i\le {n-1}\,,$$ $$(\rho_n - \rho')f_i = -a_{2i}f_i, \; i\ge n\,,$$ $$(\rho_n - \rho')e_i = -a_{2i-1}e_i, \; n\le i\le 2n-1\,,$$
$$(\rho_n - \rho')e_i = (a_{i} - a_{2i-1})e_i, \; i\ge 2n\,.$$
Hence,
\be\|\rho_n - \rho'\|_1 \le \sum_{i\ge n} a_{2i} + \sum_{n\le i\le 2n-1} a_{2i-1}
+\sum_{i\ge 2n} |a_{i} - a_{2i-1}| \le 4 \sum_{i\ge n} a_i \to 0\,,
\ee
as $n\to \infty$. Thus $\rho' \in \overline \cO_\zr  $. However $\rho'\notin \cO_\zr  $ since $\rho'$ does not have eigenvalue 0.

"{\bf If}": Let $\rho$ be finite rank, $\rho_n = U_n^*\rho U_n$ and $\|\rho_n-\rho'\|_1\to 0$ as $n\to \infty$. We need to show that $\rho'\in \cO_\zr  $.
Since rank is lower semi-continuous function (\cite{rank}), $\rho'$ has finite rank. Hence it will be sufficient to show that $\rho'$ has the same eigenvalues of the same multiplicities as $\rho$. Since $\|\rho_n-\rho'\|_1\to 0$ implies $\|\rho_n-\rho'\|\to 0$ and the set of invertible operators is open in the norm topology, for any $\lambda$ which is not eigenvalue of $\rho'$ we obtain that it is also not eigenvalue of $\rho$. Conjugating  $\rho_n-\rho'$ by $U_n^*$ we obtain that $\|U_n\rho' U_n^* - \rho\|\to 0$ and hence any $\lambda$ which is not eigenvalue of $\rho$ is also not eigenvalue of $\rho'$. Thus eigenvalues of $\rho$ and $\rho'$ are the same. Let $\lambda$ be an eigenvalue of $\rho$ and $\rho'$ and $\chi_{\lambda}$ be any continuous function which is equal 1 at $\lambda$ and $0$ at all other eigenvalues. Then $\chi_{\lambda}(\rho_n)$ and  $\chi_{\lambda}(\rho')$ are the spectral
projections of $\rho_n$ and $\rho'$ respectively, corresponding to $\lambda$. By continuity of functional calculus $$\chi_{\lambda}(\rho_n) \to
\chi_{\lambda}(\rho') $$ as $n\to \infty$. Hence, for $n$ large enough,
$$\|\chi_{\lambda}(\rho_n) -
\chi_{\lambda}(\rho')\|< 1,$$
which implies (see e.g. \cite{Davidson}) that these projections are unitarily equivalent and hence their dimensions coincide. Thus multiplicities of all eigenvalues of $\rho$ and $\rho'$ coincide.
\end{proof}

\section{Geometry and topology of quantum states}
Among the coadjoint orbits $\cO_\zr$ of the unitary group in $\u_*(\cH)$ the most important in Quantum Mechanics is clearly the orbit $\P(\cH)=\{\zr_\psi\,|\, \psi\in\cH^\ti\}$, where $\cH^\ti=\cH\setminus\{ 0\}$, consisting of \emph{pure quantum states}.

They are defined as elements $\zr$ of $\h^1(\cH)$ which are rank-one projectors. In the Dirac notation they are of the form
\be\label{pure}\zr_\psi=\frac{\kb{\psi}{\psi}}{\nm{\psi}^2}\,,
\ee
where $\psi\in\cH$ is a non-zero vector. Note that $\P(\cH)$ can be identified with the \emph{Hilbert projective space}, i.e.\   the set $\mathbb{P}\cH=\cH^\ti/\C^\ti$ of orbits in $\cH^\ti$ of the canonical action of the multiplicative group $\C^\ti=\C\setminus\{ 0\}$ of complex numbers.

As we will see in the next section, all the topologies on $\P(\cH)$: the one induced from $u_*(\cH)$ and the quotient topologies on $\cH^\ti/C^\ti$ and $\U(\cH)/\U_\zr(\cH)$, where $\U_\zr(\cH)=\{ U\in \U(\cH)\, |\, U\zr U^\dag=\zr\}$ is the isotropy subgroup of a pure state $\zr\in\P(\cH)$, coincide.

Actually, according to the general theory \cite{OR,Ratiu}, the set $\P(\cH)$ of pure states is an orbit which is a (real Banach) \emph{embedded submanifold} of $\u_*(\cH)$ whose symplectic structure is \emph{strong}. The canonical projection
\be\label{projective}
\P:\cH^\ti\rightarrow\P(\cH)\simeq \cH^\ti/\C^\ti\,,\quad \psi\mapsto [\psi]\,,
\ee
induces on $\P(\cH)$ actually a complex Hilbert manifold structure. This means that from the geometric point of view the situation is as good as it could be.

The space $\sT_{\zr_\psi}\P(\cH)$, tangent to $\P(\cH)$ at $\zr_\psi$, can be identified with $\langle\psi\rangle^\perp$. If we view $\P(\cH)$ as an $\U(\cH)$-orbit in $\u_*(\cH)$, $\sT_{\zr_\psi}\P(\cH)$ considered as a vector subspace of $\u_*(\cH)$ consists of vectors of $\u_*(\cH)$ which are commutators $[T,\zr_\psi]=T\circ\zr_\psi-\zr_\psi\circ T$, for $T$ in the Lie algebra $\u(\cH)$, i.e.\   for $T=iA$, where $A$ is Hermitian.
As
$$[T,\zr_\psi]\sim \kb{T\psi}{\psi}+\kb{\psi}{T\psi}\,,
$$
and $T\psi=i\lambda\psi +\phi$ for some $\zl\in\R$ and $\phi\perp \psi$, we have
\be\label{tP}
\sT_{\zr_\psi}\P(\cH)=\left\{\phi_\psi\,:\,\phi\perp\psi\right\}\,,
\ee
where $\phi_\psi=\kb{\phi}{\psi}+\kb{\psi}{\phi}$.
In this realization of the tangent space, for $\za\in\C^\times$, we identify $\phi_{\za\psi}$ with $(\ol{\za}\phi)_\psi$, and the complex structure $J$ on $\P(\cH)$ is represented by the map
\be\label{complex}
J\left(\phi_\psi\right)=(i\phi)_\psi=
i\left(\kb{\phi}{\psi}-\kb{\psi}{\phi}\right)\,.
\ee
Moreover,
\be\label{projection}
\sT_\psi\P(\phi)=\frac{1}{\nm{\psi}^2}(\phi^\perp)_\psi\,,
\ee
where
$$\phi^\perp=\phi-\frac{\bk{\psi}{\phi}}{\bk{\psi}{\psi}}\psi$$
is the part of $\phi$ orthogonal to $\psi$ in the decomposition $\cH=\langle\psi\rangle\op\langle\psi\rangle^\perp$.

As a coadjoint orbit, $\P(\cH)$ carries the canonical symplectic structure $\oP$ given by the Kostant-Kirillov-Souriau formula (cf.\    (\ref{PB0}))
$$\oP([T,\zr_\psi],[T',\zr_{\psi}])=\langle\zr_\psi,[T,T']\rangle=\tr(\zr_\psi[T,T'])\,.
$$
Assume for a moment that $\nm{\psi}=1$. Then,
$$\tr(\zr_\psi[T,T'])=\bk{\psi}{(TT'-T'T)(\psi)}=\bk{T'\psi}{T\psi}-\bk{T\psi}{T'\psi}
=-2\im\bk{T\psi}{T'\psi}\,,
$$
so that
$$\oP([T,\zr_\psi],[T',\zr_{\psi}])=-\frac{2}{\nm{\psi}^2}\im\bk{T\psi}{T'\psi}\,.
$$
Since $[T,\zr_\psi]=(T\psi)_\psi/\nm{\psi}^2$, we get
\be\label{symplectic}
\zw_{\P(\cH)}(\phi_\psi,\phi'_\psi)=-2\im (\bk{\phi}{\phi'})\nm{\psi}^2\,.
\ee
It is easy to see now, that
$$g_{\P(\cH)}(\phi_\psi,\phi'_\psi)=\oP(\phi_\psi,J(\phi'_\psi))$$
is a Riemannian tensor,
\be\label{riemannian}
g_{\P(\cH)}(\phi_\psi,\phi'_\psi)=2\re(\bk{\phi}{\phi'})\nm{\psi}^2=\tr(\phi_\psi\phi'_\psi)\,,
\ee
i.e.\   these three structures, $(\zw_{\P(\cH)},g_{\P(\cH)},J)$ turn $\P(\cH)$ into a Hilbert-K\"ahler manifold. For the theory of (finite-dimensional) K\"ahler manifolds we refer to \cite{Mor} and references therein.
Note that the Riemannian metric on $\P(\cH)$ is induced from the $\L^2$-metric on $u_*(\cH)$. Of course,
as $\sT_{\zr_\psi}\P(\cH)$ is a linear space of rank$\le 2$ operators, the $\L^2$-topology coincides with the $\L^1$-topology induced from $\u_*(\cH)$.  The hermitian product on $\sT_\zr\P(\cH)$
\be\label{hermitian}
\bk{\phi_\psi}{\phi'_\psi}_{\P(\cH)}=g_{\P(\cH)}(\phi_\psi,\phi'_\psi)-
i\zw_{\P(\cH)}(\phi_\psi,\phi'_\psi)=2\bk{\phi}{\phi'}\nm{\psi}^2\,.
\ee
It is easy to see that the action of the unitary group $\U(\cH)$ on $\cH$ projects to an action $U\mapsto U_\P$ on the Hilbert projective space, $U_\P[\psi]=[U(\psi)]$. The tangent map acts as
\be\label{U-action}
\sT U_\P(\phi_\psi)=U(\phi)_{U(\psi)}\,.
\ee
This action consists of automorphisms of the K\"ahler structure, i.e.\   of automorphisms of all the three structures,
$\zw_{\P(\cH)}, g_{\P(\cH)}, J$. Of course, as any two of these structures determine the third one, preserving two of these structures implies preserving of the third one. Note that $U_\P=\Id$ if and only if $U=e^{i\lambda}\Id$ for some $t\in\R$, so the effective action is provided by the \emph{projective unitary group} $\PU(\cH)=\U(\cH)/S^1$, where the canonical realisation of $S^1$ as a normal subgroup of $\U(\cH)$ is $e^{it}\mapsto e^{it}\Id$. This action gives all automorphisms of the K\"ahler structure on $\P(\cH)$.
\begin{theorem}\label{tx}
Any automorphism of the K\"ahler structure on $\P(\cH)$ is of the form $U_\P$ for some $U\in\U(\cH)$.
\end{theorem}
In the proof we will use the following lemma from Riemannian geometry we learned from Jason DeVito.
\begin{lemma}
Let $M$ be a connected Riemannian (Banach) manifold and $f:M\to M$ be an isometry. Suppose that there is $p\in M$ such that $f(p)=p$ and the derivative $D_pf$ of $f$ at $p$ is the identity on $\sT_pM$. Then, $f$ is the identity.
\end{lemma}

\begin{proof}
Put
$$X=\left\{ q\in M\,|\, f(q)=q\,,\quad D_qf=\Id\right\}\,.
$$
The set $X\subset M$ is non-empty and closed. It suffices to show that it is open.
For, let $q\in X$ and let $W$ be a normal neighbourhood of $q$  on which the inverse of the exponential map
$\exp_q$ is well defined. Since,
$$f(\exp_q(tv))=\exp_{f(q)}(tD_gf(v))\,,
$$
any $r\in W$ is a fixed point of $f$ and so $D_rf=\Id$.
\end{proof}
\textit{Proof of Theorem \ref{tx}.} Let $f:\P(\cH)\to\P(\cH)$ be an automorphism of the K\"ahler structure and $f([\psi])=[\psi']$. Fixing $\psi,\psi'$ chosen such that $\nm{\psi}=\nm{\psi'}=1$,
we can view $D_{[\psi]}f$ as a map
$$D_\psi^{\psi'}f:\langle\psi\rangle^\perp\to\langle\psi'\rangle^\perp\,.$$
As $f$ preserves all the structures, we have, according to (\ref{complex}) and (\ref{hermitian}),
$$D_\psi^{\psi'} f(i\phi)=iD_\psi^{\psi'} f(\phi)\,,\quad
\bk{D_\psi^{\psi'} f(\phi)}{D_\psi^{\psi'} f(\phi')}=\bk{\phi}{\phi'}\,,
$$
thus $D_\psi^{\psi'} f$ is unitary. It is easy to see now that the unique unitary operator $U$ which maps $\psi$ to $\psi'$ and equals $D_\psi^{\psi'}f$ on $\langle\psi\rangle^\perp$ induces on $\P(\cH)$ the map $U_\P$ such that $U_\P([\psi])=[\psi']$ and $D_{[\psi]}U_\P=D_{[\psi]}f$.
In view of the above lemma, $f=U_\P$.
\epf
\begin{remark}\label{re1}
A vector field $Y$ on a finite-dimensional complex manifold which is an infinitesimal automorphism of the complex structure $J$, i.e.
$$[Y,J(Z)]=J[Y,Z]\quad\text{for each vector field}\quad Z\,, $$
we call \emph{real holomorphic}.
Real holomorphic vector fields can be equivalently characterized as vector fields $Y$ that induce holomorphic flows, or such vector fields for which
$Y-iJ(Y)$ are holomorphic (cf.\    \cite[Lemma 8.7]{Mor}).
Infinitesimal automorphism of a K\"ahler structure can be then described as real holomorphic Hamiltonian vector fields, or equivalently, as real holomorphic Killing vector fields or Hamiltonian-Killing vector fields.
In infinite dimensions the situation is more complicated, as generators of strongly-continuous groups of automorphism might be only densely defined. We come to these questions in section \ref{holohamiltononians}.
\end{remark}

\medskip
The closed convex hull $\D(\cH)$ of $\P(\cH)\simeq\mathbb{P}$ in $\u_*(\cH)$ is the set of all \emph{(mixed) quantum states}. They are positive semi-definite operators from $\L^1(\cH)$ with trace 1. Being trace-class operators, they have the form
\be\label{mixed}
\zr=\sum_{n=1}^\infty\zl_n\zr_{\psi_n}\,,\quad \zl_n>0\,,\quad \sum_{n=1}^\infty\zl_n=1\,,
\ee
for a system $(\psi_n)$ of orthogonal vectors from $\cH^\ti$.

Note that the set $\P(\cH)$ of pure states can be viewed as a subset of each of $\L^p(\cH)$, however the trace is defined on the closed convex hull $\D(\cH)$ only in $\L^1(\cH)$. Moreover, in infinite dimension,
the closure of the convex hull of $\P(\cH)$ in $\L^p(\cH)$, $p>1$, contains 0 (which is clearly not the case for $p=1$):
$$\norm{\frac{\zr_{e_1}+\cdots+\zr_{e_n}}{n}}_p^p=n^{1-p}\xrightarrow[n\to\infty]{}\,0\,.
$$
Here, $e_1,\dots,e_n$ is an orthonormal sequence of vectors from $\cH$.

\medskip
As  $\D(\cH)$ is contained in the closed unit ball in $L^1(\cH)$ and this ball is compact in the weak$^*$ topology (\emph{Alaoglu's Theorem}), so in our case the weak topology induced from the space $\L^\infty(\cH)$ of compact operators, $\D(\cH)$ is compact in the weak$^*$ topology, so, according to the \emph{Krein-Milman Theorem}, it is the closure of the convex hull of its \emph{extreme points}. It is then easy to see that the extreme points are just pure quantum states (see (\ref{mixed}).

The boundary of $\D(\cH)$ consists those mixed states $\zr$ whose kernel, $\ker(\zr)$, is non-trivial.

\begin{remark}
Note that for a finite-dimensional $\cH$, the identification $\P(\cH)\simeq\U(\cH)/\U_\zr(\cH)$  is trivial, but working with infinite-dimensional Hilbert space brings a new topological flavour. For instance, the unit sphere $S^\infty(\cH)=\{\psi\in\cH\, |\, \nm{\psi}=1\}$ in an infinite-dimensional separable Hilbert space $\cH$, thus the unitary group $\U(\cH)$ itself, is contractible in the norm topology \cite{Kui}, which is clearly false in the finite-dimensional case. In fact, in infinite dimension $S^\infty(\cH)$ is even analytically diffeomorphic with the Hilbert space itself \cite{bessaga66,dobrowolski95}. Note also that any (real) smooth Hilbert manifold is diffeomorphic to an open set of the real Hilbert space \cite{EE}.

Of course, although the fundamental group of the Hilbert projective space vanishes, $\zp_1(\P\cH)=0$, the projective space $\mathbb{P}=\P(\cH)$ is no longer contractible in infinite dimension, since the fiber of the \emph{Serre fibration} $S^\infty(\cH)\to \mathbb{P}$ is the circle $S^1=\{ z\in\, |\, \nm{z}=1\}$, and the corresponding \emph{long exact sequence of homotopy groups} looks like
$$\dots\ra\zp_n(S^1)\ra\zp_n(S^\infty)\ra\zp_n(\mathbb{P})\ra\zp_{n-1}(S^1)\ra\dots\ra\zp_0(S^\infty)\ra 0\,.$$
In particular, we get
$$\dots\ra\zp_2(S^\infty)\ra\zp_2(\mathbb{P})\ra\zp_1(S_1)\ra\zp_1(S^\infty)\ra\dots\,.
$$
Since contractibility implies $\zp_k(S^\infty)=0$ for $k>0$, we end up with the exact sequence
$$0\ra\zp_2(\mathbb{P})\ra\zp_1(S_1)\ra 0\,,$$
which shows that $\zp_2(\mathbb{P})$ is isomorphic with $\zp_1(S^1)=\Z$, so $\mathbb{P}$ is not contractible.
\end{remark}

\subsection{Local embedding of pure states into the unitary group}
Of course it would be interesting to have, at least locally,  an explicit map $\cO_\zr   \to \U(\cH)$, composition of which with the quotient map $\U(\cH) \to \U(\cH)/\U_\rho(\mathcal H) $ would give a homeomorphism (here $\U_\rho(\mathcal H)$ is the stabilizer of $\rho$).
This would imply that the topology in the orbit $\cO_\zr$ of finite rank element $\zr$  coincides with the quotient topology $\U(\cH)/\U_\rho(\mathcal H)\simeq \cO_\zr$. We will construct such a map.

Let $\rho$ be  a state of finite rank  and $\lambda_1, \ldots, \lambda_k$ be its non-zero eigenvalues. Let $P_1, \ldots, P_k$ be the orthogonal projections on the corresponding eigenspaces. Let $P_0$ be the orthogonal projection on the kernel of $\rho$ and let $\lambda_0 = 0. $
Let $\rho' \in \cO_\zr  $. Then it has the same eigenvalues $\lambda_0, \lambda_1, \ldots, \lambda_k$. Let $Q_0, Q_1, \ldots, Q_k$ be the orthogonal projections on the corresponding eigenspaces. Since the spectra of $\rho, \rho'$ are finite, the projections $P_i, Q_i$ are smooth (in fact, polynomial) functions of $\rho$ and $\rho'$ respectively
 \begin{equation}\label{functions} P_i = f_i(\rho), \; Q_i = f_i(\rho').\end{equation} It implies that if $\rho$ and $\rho'$ are sufficiently close then $P_i, Q_i$ are sufficiently close. In particular we can ensure that $$\sum_{i=0}^k\|P_i - Q_i\| \le 1/2$$  (here $\|\;\|$ is the operator norm). Now we are going to construct a unitary $U$  which would conjugate $\rho$ and $\rho'$. For that we will use a construction from (\cite{Davidson}, Lemma III.3.2) where it is proved that sufficiently close projections are unitarily equivalent.

Let $$X = \sum_{i=0}^k Q_iP_i.$$ Then $$X^*X = \sum_{i=0}^k P_iQ_iP_i = \sum_{i=0}^k P_i + \sum_{i=0}^k P_i(Q_i-P_i)P_i \ge \Id - \sum_{i=0}^k \|Q_i - P_i\| \Id = \frac{1}{2} \Id.$$ Similarly $XX^* \ge \frac{1}{2} \Id.$ Hence $X$ is invertible.  We also have \begin{equation}\label{commute}Q_iX = XP_i\end{equation} and $$P_i X^*X = X^*X P_i,$$ for all $i's$. The latter equality implies that \begin{equation}\label{commute2}P_i |X| = |X| P_i,\end{equation} for all $i's$.

Let $U$ be the unitary from the polar decomposition $X = U |X|$ of $X$. By (\ref{commute}) and (\ref{commute2}) we obtain
$$UP_i  = X|X|^{-1}P_i = XP_i|X|^{-1} = Q_iX|X|^{-1} = Q_i U.$$ Since $\rho = \sum_{i=0}^k \lambda_i P_i$, $\rho' =  \sum_{i=0}^k \lambda_i Q_i$,
 we conclude that $U\rho = \rho' U$. Thus $\rho' = U\rho U^*$.

Now the mapping $\Phi:\rho' \mapsto U$ gives us a local embedding of $\cO_\zr  $ into $\U(\cH)$. It follows  from (\ref{functions}) and the construction of $U$ that this mapping is continuous with respect to the norm topology. Since $\rho$ has finite rank, on $\cO_\zr  $ the norm topology is equivalent to the topology of the space of trace-class operators and thus we have a continuous local embedding of $\cO_\zr  $ into $\U(\cH)$. This proves the following.

\begin{theorem} The above constructed map $\Phi$ establishes an embedding of an open neighbourhood $W$ of $\zr$ in the orbit $\cO_\zr$ into the unitary group $\U(\cH)$ such that the composition of $\Phi$ with the canonical projection $\U(\cH) \to \U(\cH)/\U_\rho(\mathcal H) $ is the identity. In consequence, the topologies in $\cO_\zr$  induced from $\u_*(\cH)$ and from $\U(\cH)/\U_\rho(\mathcal H) $ coincide.

\end{theorem}

\section{Quantum dynamics in the Tulczyjew picture}
\subsection{Quantum Tulczyjew triple}
Our aim in this section is to construct a quantum analog of the Tulczyjew triple (\ref{tt}). First, notice that the Hilbert space $\cH$ can be viewed as an infinite-dimensional K\"ahler manifold with the standard complex structure $J$ the Riemannian and symplectic structures $\gh$ and $\oh$ being the real and (minus) imaginary parts of the hermitian product,respectively, \cite{ashtekar96,grabowski05}:
\be J(x)=i\cdot x,\qquad \gh(x,y)-i\cdot\oh(x,y)=\bk{x}{y}\,.
\ee
It is clear that the unitary group acts by symplectomorphisms which are simultaneously isometries
of the Riemannian structure. This is the starting point for developing a geometric quantum mechanics \cite{CM,cirelli90,cirelli91,cirelli03}.

To make a connection with the classical triple we will view the symplectic manifold $(\cH,\oh)$ as the cotangent bundle $\sT^*H$ with the canonical symplectic form $\zw_{\sT^*H}$, where $H$ is a `real part' of $\cH$, i.e.\   a real vector subspace of $\cH$ such that $\cH=H\oplus_\R iH$. In what follows, with some abuse of notation, we will understand the symbol $H\oplus H'$ as $H\oplus_\R H'$ if $H,H'$ are interpreted as real spaces, and $\cH\oplus\cH'$ as $\cH\oplus_\C\cH'$ if $\cH,\cH'$ are complex.  Note that $H$ and $iH$ are real Hilbert spaces with the scalar product $g_H$ induced from $\gh$, and that one can view $iH$ (\emph{via} $\oh$) as the dual space $H^*$ of $H$. Of course, the scalar product $g_H$ induces a canonical isomorphism $H\simeq H^*=iH$. In general there is no canonical choice of the `real part' $H$. We can span $H$ by an orthonormal basis $(e_k)$ in $\cH$ and real coefficients:
$$H=\left\{ q^ke_k\,:\, q_k\in\R\,,\quad \sum_k|q^k|^2<\infty\right\}\,,
$$
and regard $q=(q^k)$ as (real) coordinates in $H$. We will also write with some abuse of notation $q\in H$. It is easy to see that in the coordinates $(q,p)$ in
$$\cH=H\oplus iH=H\oplus H^*$$
for which $x\in\cH$ is written as
$x=\sum_k(q^k+ip_k)e_k$, the symplectic form $\oh$ reads
\be
\oh=\xd p_k\we\xd q^k\,,
\ee
i.e.\   $\oh$ coincides with $\zw_{\sT^*H}$.
\begin{remark}
Let us note that in some situations the choice of $H$ can be canonical. For instance,
if $\cH=\L^2(\Omega,\zm)$ is the space of square-integrable function on a measure set $(\Omega,\zm)$, then it is natural to choose $H$ as the space $\L^2_\R(\zW,\zm)$ of real functions from $\cH=\L^2(\Omega,\zm)$ (cf.\    \cite{MV}).
\end{remark}
Having chosen $H$ as the configuration manifold, we can write the corresponding Tulczyjew triple (the only difference is that we work with infinite-dimensional real Hilbert manifolds):
\be\label{qtt}
    \xymatrix@R-4mm @C-10mm
        { & \sT^\*\sT^\*H \ar[ldd]_*{} \ar[rd]^*{} & & & \sT \sT^\* H \ar[rrr]^*{{\za}_H}
        \ar[lll]_*{\zb_H} \ar[ldd]^*{} \ar[rd]^*{}& & & \sT^\*\sT H \ar[ldd]^*{} \ar[rd]^*{} & \cr
        & & \sT H \ar[ldd]^*{} & & & \sT H  \ar[ldd]^*{} \ar@{=}[lll]^*{} \ar@{=}[rrr]^*{} & & & \sT H \ar[ldd]^*{}  \cr
         \sT^\* H  \ar[rd]^*{}  & & & \sT^\* H \ar@{=}[rrr]^*{} \ar@{=}[lll]^*{} \ar[rd]^*{}& & & \sT^\* H
        \ar[rd]^*{} & & \cr
        & H  & & &  H \ar@{=}[rrr]^*{} \ar@{=}[lll]^*{} & & & H & }.
\ee
We have already identified $\sT^*H$  with $\cH=H\oplus H^*$. As $H$ is a linear Hilbert manifold, we have canonical identifications
$$\sT\sT^*H=\cH\oplus\cH\,,\quad \sT^*\sT^*H=(H\oplus H^*)\oplus(H^*\oplus H)\,,\quad  \sT^*\sT H=(H\oplus H)\oplus(H^*\oplus H^*)\,.$$
We can choose the corresponding coordinates such that the isomorphisms $\za_H$ and $\zb_H$ take the form of the identities:
\bea\label{coor}
(q,p,\dot q,\dot p)&&\quad \text{on}\quad (H\oplus H^*)\oplus(H\oplus H^*)=\sT\sT^*H\,,\\
(q,\dot q,\dot p, p)&&\quad \text{on}\quad (H\oplus H)\oplus(H^*\oplus H^*)=\sT^*\sT H\,,\label{coor1}\\
(q,p,\dot p,-\dot q)&&\quad \text{on}\quad (H\oplus H^*)\oplus(H\oplus H^*)=\sT^*\sT^*H\,.\label{coor2}
\eea
Identifying $\sT\sT^*H$ with $\cH\oplus\cH$ \emph{via}
$$(q,p,\dot q,\dot p)\mapsto (x,\dot x)=(q+ip,\dot q+i\dot p)\,,$$
we can write the canonical symplectic form
\be\label{0}
\zw_{0+}=\zw_{\sT\sT^*H}=\xd \dot p_k\we\xd q^k+\xd p_k\we\xd \dot q^k
\ee
on $\sT\sT^*H$, being the tangent lift $\dt \zw_{\sT^*H}$ of the canonical symplectic form $\zw_{\sT^*H}$ on $\sT^*H$, as
\be\label{sf}
\zw_{0+}\left((x,\dot x),(y,\dot y)\right)=-\im\bk{(x,\dot x)}{(y,\dot y)}_{0+}\,,
\ee
where the pseudo-Hermitian form $\bk{\cdot}{\cdot}_{0+}$ (we will see that it is actually Hermitian) reads
\be\label{bk0}
\bk{(x,\dot x)}{(y,\dot y)}_{0+}=\bk{\dot x}{y}+\bk{x}{\dot y}\,.
\ee
On $\sT^*\sT H$ the canonical symplectic form $\zw_0=\zw_{\sT^*\sT H}$ in coordinates (\ref{coor1}) reads exactly as $\zw_{0+}$ in (\ref{sf}) (the map $\za_H$ is a symplectomorphism). Writing $Q$ for $(q,\dot q)\in H\oplus H$ and $P$ for $(\dot p,p)\in H^*\oplus H^*$ we can also write
\be\label{zw0}
\zw_0\left((Q,P),(Q',P')\right)=\langle P',Q\rangle-\langle P,Q'\rangle\,,
\ee
where $\langle\cdot,\cdot\rangle$ is the canonical pairing between $H^*\oplus H^*$ and $H\oplus H$.

Actually, using the identification $H\simeq H^*$, we can identify also the bundles $\sT^*\sT H$ and $\sT^*\sT^*H$ with $\cH\oplus\cH$, but the identifications are far from being canonical. Also
the expression similar to (\ref{sf}) for the canonical symplectic forms depends on the chosen identification. In this sense, all main objects in the quantum Tulczyjew triple can be viewed as $\cH\oplus\cH$. However, we will not exploit explicit identifications (except this for the tangent bundle $\sT\sT^*H$ which contains the quantum dynamical part), as for generating the dynamics we will use anyhow the canonical Tulczyjew isomorphisms $\za_H$ and $\zb_H$.

Note only that on $\cH_1=\cH\oplus\cH$ we have two other pseudo-Hermitian structures:
\be\label{hi-her}
\bk{(\zf,\zf')}{(\zc,\zc')}_\pm=\bk{\zf'}{\zc'}\pm\bk{\zf}{\zc}\,.
\ee
and the corresponding symplectic structures
\be\label{zw1}
\zw_\pm\left((\zf,\zf')(\zc,\zc')\right)=-\im\left(\bk{(\zf,\zf')}{(\zc,\zc')}_\pm\right)\,.
\ee
The structure $\bk{\cdot}{\cdot}_+$ is actually the canonical symplectic structure on $\cH_1=\cH\oplus\cH$.
We will show later on that the symplectic structures $\zw_0$ and $\zw_+$ on $\cH_1=\cH\oplus\cH$ are actually linearly equivalent, i.e.\   there is a complex linear isomorphism (\emph{Cayley map}) $C:\cH_1\to \cH_1$ which maps $\zw_+$ onto $\zw_{0+}$.

\subsection{Self-adjoint and anti self-adjoint relations}
Let us recall that in the Tulczyjew approach an implicit dynamics is a Lagrangian submanifold in $\sT\sT^*H$. In the quantum case we will consider only \emph{complex linear Lagrangian submanifolds}, i.e.\  those Lagrangian submanifolds $V\subset\sT\sT^*H$ which, up to the above identification, are complex linear subspaces of the Hilbert space $\cH\oplus\cH$.
The linear subspace $V$ of $\cH\oplus \cH\simeq \cH\times \cH$ can be
understood as a (linear) relation in $\cH$. In particular, if $A:
\cH\supset\mathcal{D}\rightarrow \cH$ is a complex linear operator in the domain
$\mathcal{D}$, its graph, $\mathfrak{G}(A)=\{(x,Ax), x\in\mathcal{D}\}$, is a linear
relation (subspace) in $\cH\oplus\cH$. Such subspaces $V\subset\cH\oplus\cH$
are characterized by the condition that the projection
$(x,x^\prime)\xrightarrow{\pi_1} x$ onto the first component is injective on
$V$. Note that an operator $A$ is called \emph{closed} if $\mathfrak{G}(A)$ is closed.

The crucial observation is the following.
\begin{theorem}
The graph $\mathfrak{G}(A)$ of a complex linear operator $A:
\cH\supset\mathcal{D}\rightarrow \cH$ is an isotropic submanifold for the symplectic structure (\ref{sf}) if and only if the operator is anti-symmetric:
$$\bk{Ax}{y}+\bk{x}{Ay}=0\quad\text{for}\quad x,y\in\cD\,.$$
Moreover, $\mathfrak{G}(A)$ is a Lagrangian submanifold if and only if $\cD$ is dense in $\cH$ and $A$ is anti self-adjoint, i.e.\   $A^\dag=-A$.
\end{theorem}
\begin{proof}
According to (\ref{sf}), $\mathfrak{G}(A)$ is syplectically isotropic if and only if
the imaginary part of $\bk{Ax}{y}+\bk{x}{Ay}$ vanishes. But $A$ is complex linear, so $(ix,iAx)$ belongs also to $\mathfrak{G}(A)$, thus the imaginary part of
$$\bk{iAx}{y}+\bk{ix}{Ay}=-i\left(\bk{Ax}{y}+\bk{x}{Ay}\right)\,,$$
which is the real part of
$\bk{Ax}{y}+\bk{x}{Ay}$,
vanishes as well. In consequence, $\mathfrak{G}(A)$ is symplectically isotropic if and only if $\bk{Ax}{y}+\bk{x}{Ay}$ vanishes for all $x,y\in\cD$, i.e.\   $A$ is anti-symmetric.
If $\mathfrak{G}(A)$ is Lagrangian, then $\cD$ must be dense in $\cH$. Indeed, if $x_0$ is orthogonal to $\cD$, then $(0,x)$ is symplectically orthogonal to $\mathfrak{G}(A)$, thus $(0,x)\in \mathfrak{G}(A)$ and $x=0$. Moreover, $(y,\dot y)$ is symplectically orthogonal to $\mathfrak{G}(A)$ if and only if
$$\bk{x}{\dot y}+\bk{Ax}{y}=0\,,$$
i.e.\   if and only if $y$ is in the domain of $A^\dag$ and $\dot y=A^\dag y$. Hence, $\mathfrak{G}(A)$ is Lagrangian if and only if $A^\dag=-A$.
\end{proof}

The above result is in full accordance with the fact that the explicite quantum dynamics is described by one-parameter subgroups of the unitary group whose generators are anti self-adjoint operators (Stone theorem). Complex linear Lagrangian submanifolds in $\cH\oplus\cH$ we will call also \emph{anti self-adjoint relations}. In our framework they will play the role of implicit quantum dynamics.

Of course, anti self-adjoint relations are in a close correspondence with
\emph{self-adjoint relations}. The latter are Lagrangian submanifolds of $\cH\oplus\cH$ equipped with the symplectic form
\be\label{sf0}
\zw_{0-}\left((x,\dot x),(y,\dot y)\right)=-\im\bk{(x,\dot x)}{(y,\dot y)}_{0-}\,,
\ee
where
$$\bk{(x,\dot x)}{(y,\dot y)}_{0-}=i\left(\bk{\dot x}{y}-\bk{x}{\dot y}\right)\,.$$
The map
$$\cH\oplus\cH\ni(x,\dot x)\mapsto(ix,\dot x)\in\cH\oplus\cH$$
is a symplectomorphism between $\zw_{0+}=\zw_{\sT\sT^*H}$ and $\zw_{0-}$.

\begin{theorem}\label{th1} A complex linear relation $V\subset\cH\oplus\cH$ is (anti) self adjoint if and only if the inverse relation
$$V^{-1}=\left\{(y,x)\in\cH\op\cH\,:\,(x,y)\in V\right\}$$
is (anti) self-adjoint. If $V$ is (anti) self-adjoint, then the orthogonal complement $V_+^\perp$ of the domain of $V$,
$$V_+=\operatorname{dom}(V)=\left\{ x\in\cH\,:\, (x,y)\in V\ \text{for some}\ y\right\}\,,$$
is the kernel of $V^{-1}$,
$$\ker(V^{-1})=\left\{ x\in\cH\,:\, (0,x)\in V\right\}\,.
$$
In particular, $V$ is a graph, $V=\mathfrak{G}(A)$,
if and only if it is densely defined. In this case, $A$ is a bounded operator if and only if $V_+=\cH$.

\end{theorem}
\begin{proof}
The first statement easily follows from the formulas (\ref{sf}) and (\ref{bk0}). It is also easy to see that if $V$ is a complex linear relation on a domain $V_+$, then any vector $(0,x)$ is $\bk{\cdot}{\cdot}_{0+}$-orthogonal to $V$, thus belongs to $V$,  if and only if $x\in V_+^\perp$.

Assume now that $V=\mathfrak{G}(A)$ and $A$ is bounded, $\nm{Ax}\le K\nm{x}$, and let $x_n\in V_+$, $x_n\to x_0$. Then, $(x_n)$ is a Cauchy sequence, $\lim_{n,m\to\infty}\nm{x_n-x_m}=0$, so
$$\lim_{n,m\to\infty}\nm{Ax_n-Ax_m}=\lim_{n,m\to\infty}\nm{A(x_n-x_m)}\le K\lim_{n,m\to\infty}\nm{x_n-x_m}=0
$$
and $(Ax_n)$ is a Cauchy sequence, $Ax_n\to y_0$.
Since $V$ is closed, $(x_0,y_0)\in V$, thus $y_0=Ax_0$. This implies that the domain of $A$ is closed, thus equals $\cH$. Conversely, if the domain of $A$ is $\cH$, then the canonical projection $pr_1:\cH\oplus\cH\to\cH$ onto the first component maps $V$ injectively onto $\cH$. In view of the Banach Inverse Theorem, the map $A$, which is the inverse of $V\to\cH$, is bounded.

\end{proof}

\medskip
General (anti) self-adjoint relations we can describe as follows.
\begin{theorem}\label{tsarel}
Any (anti) self-adjoint relation $V\subset\cH\oplus\cH$ is of the form
\be\label{sarel}
V_A=\{ (x,Ax+v)\,|\, x\in\cD\,,\quad v\in\cD^\perp\}\,,
\ee
where $A:\cH\supset\cD\to\cH$ is an (anti) self-adjoint operator densely defined in the
Hilbert space $\ol{\cD}\subset\cH$.
\end{theorem}
\begin{proof}
Let $\cH_0=\ker(V^{-1})$ and $\cH_1=\cH_0^\perp$. It is easy to see that
$$V=\{(x,y+v)\,|\,(x,y)\in V_1\,,\quad v\in H_0\}\,,
$$
Where $V_1=V\cap(H_1\oplus H_1)$ is the restriction of $V$ to $H_1\oplus H_1$. As $V_1$ is clearly Lagrangian (anti) self-adjoint and densely defined, it is a graph of an (anti) self-adjoint operator densely defined on $\cD\subset\cH_1$.
\end{proof}

\subsection{Quantum dynamics in the Tulczyjew picture}
The procedure of generating a quantum dynamics will be now standard in the Tulczyjew picture.
Since we intend to work only with linear relations, we will consider for simplicity only Lagrangians $L$ which are real functions, quadratic in $Q=(q,\dot q)$, and defined in domains of their differentiability $\cD_0$ in $\sT H=H\oplus H$ which are linear subspaces of $\sT H=H\oplus H$,
\be\label{L}
L:\sT H\supset\cD_0\to\R\,.
\ee
This means that on $\cD_0$ the Lagrangian is defined together with its differential
$$\xd L:\cD_0\to \ol{\cD_0}^{\,*}\,,$$
associating with the points $Q$ of $\cD_0$ linear functionals $\xd L(Q)$ which are continuous on $\cD_0$, thus on the closure $\ol{\cD_0}$ of $\cD_0$. Of course, $\xd L(Q)(Q')$, for $Q'\in\cD_0$ is the directional derivative of $L$ in the direction of $Q'$.

We regard $\cD_0$ as a vakonomic constraint and generate the Lagrangian submanifold $S(L)$ in $\sT^*\sT H$. Actually, in the infinite dimensions we construct $S(L)$ as
the closure of the linear subspace
$$S(L)^0=\left\{(Q,P)\in\sT^*\sT H\,:\,Q\in\cD_0\,,\quad pr(P)=\xd L(Q)\right\}\,,
$$
where $pr:H^*\oplus H^*\to \ol{\cD_0}^{\,*}$ is the canonical projection, dual to the embedding $\ol{\cD_0}\hookrightarrow H^*\oplus H^*$.
\begin{theorem}\label{th2}
For any Lagrangian (\ref{L}), defined on a linear subspace $\cD_0$ of $\sT H=H\oplus H$ together with its differential and quadratic in $Q=(q,\dot q)$, the linear subspace $S(L)\subset \sT^*\sT H=(H\oplus H)\oplus(H^*\oplus H^*)$ is Lagrangian with respect to the canonical symplectic form $\zw_0$.
\end{theorem}
\begin{proof}
Let $g_0$ be the canonical scalar product on $\sT H=H\oplus H$,
$$g_0(Q,Q')=g(q,q')+g(\dot q,\dot q')\,.
$$
That $L$ is quadratic on the domain $\cD_0$ of differentiability means that there is a (real) linear operator $B:{\cD_0}\to\ol{\cD_0}$ such that $B$ is $g_0$-symmetric, i.e.\   $g_0(Q',BQ)=g_0(BQ',Q)$ for all $Q,Q'\in\cD_0$, and
$$L(Q)=\frac{1}{2}g_0(Q,BQ)\,,\quad Q\in\cD_0\,.
$$
Hence, $\xd L(Q)(Q')=g_0(BQ,Q')$ for $Q\in\cD_0$, $Q'\in\ol{\cD_0}$ and it is easy to see that $S(L)$ is $\zw_0$-isotropic. According to (\ref{zw0}), $(Q',P')\in (H\oplus H)\oplus(H^*\oplus H^*)$ is $\zw_0$-orthogonal to all $(Q,P)$, with $Q\in\cD_0$ and $pr(P)=\xd L(Q)$,
if and only if
$$\langle P',Q\rangle=\langle P,Q'\rangle\,,\ \text{for all}\ Q\in\cD_0\,,\ pr(P)=\xd L(Q)\,.
$$
Since $pr(P_1)=\pr(P_2)$ means that $P_1-P_2$ belongs to the annihilator $\cD_0^o\subset\,H^*\oplus H^*$ of $\cD_0$, it follows that
$$\langle P,Q'\rangle=0\,,\ \text{for all}\ P\in\cD_0^o\,,$$
i.e., $Q'\in\ol{\cD_0}$. Hence,
$$pr(P')(Q)=\langle P',Q\rangle=\langle P,Q'\rangle=\xd L(Q)(Q')=g_0(BQ,Q')\,,\ \text{for all}\ Q\in\cD_0\,.$$
Since $Q'$ is in the closure of $\cD_0$, there is a sequence $(Q_n)$ in $\cD_0$ such that $Q_n\rightarrow Q'$. Therefore,
$$pr(P')(Q)=g_0(BQ,Q')=\lim_{n\to\infty}g_0(BQ,Q_n)=\lim_{n\to\infty}g_0(Q,BQ_n)\,,
$$
i.e., $g_0(\cdot,BQ_n)\rightarrow pr(P')$ in $\ol{\cD_0}^{\,*}$. Now, we can choose $P_n\in H^*\oplus H^*$ such that $(Q_n,P_n)\in S(L)^0$ and  $(Q_n,P_n)\rightarrow(Q',P')$, thus $(Q',P')$ is in the closure of $S(L)^0$, so this closure, $S(L)$, is a Lagrangian submanifold.

\end{proof}

Now, \emph{via} the symplectomorphism $\za_H$, we view $S(L)$ as a Lagrangian submanifold $V(L)$ in $(\sT\sT^*H,\zw_{0+})$. If we assume that $V(L)$ is a complex linear relation, it is anti self-adjoint and represents the implicit quantum dynamics. According to Theorem \ref{th1}, the kernel of $V(L)^{-1}$ is the orthogonal complement of $\cD=V(L)_+=\operatorname{dom}(V(L))$.,

The first integrability extract (cf.\    (\ref{ie})),
$$V(L)^1=V(L)\bigcap \left(\cD\oplus\ol{\cD}\right)\,,
$$
is now the graph of an anti self-adjoint operator $-\frac{i}{\hbar}A$ defined on the domain $\cD$ which is dense in the closed subspace $\ol{\cD}=\ol{V(L)_+}$ of $\cH$ representing the Hamiltonian constraint, $V(L)^1=\mathfrak{G}(A)$.
The operator $A$ is the Schr\"odinger operator in the Schr\"odinger picture.

The dynamics is explicit on $\ol{\cD}$ and generates, in view of the Stoke's Theorem, a one-parameter group $\exp(-\frac{it}{\hbar}A)$ of unitary transformations of the phase space $\ol{\cD}\subset\sT^*H$. The projections of the trajectories to $H$ are solutions of the corresponding Euler-Lagrange equations.

The implicit dynamics $V(L)$  comes also from the classical constrained Hamiltonian
\be\label{H}
h_A:\sT^* H\simeq\cH\supset\cD\to\R\,,\quad h(x)=\frac{1}{2\hbar}\bk{x}{Ax}\,,
\ee
defined on $\cD$. Indeed, the symplectic structure on $\cH=\sT^*H$ is the minus of the imaginary part of the Hermitian structure, so for $x,y\in\cD$ we have
$$\xd h_A(x)(y)=\frac{1}{2\hbar}\left(\bk{y}{Ax}+\bk{x}{Ay}\right)=
\frac{1}{\hbar}\mathcal{R}e\bk{y}{Ax}=\zw_{0+}\left(y,\frac{-iAx}{\hbar}\right)\,.
$$
Hence, the constrained Hamiltonian dynamics is represented by the Lagrangian submanifold
\be\label{Vh}V(h_A)=V_{-iA/\hbar}=\left\{(x,-\frac{i}{\hbar}Ax+v)\,:\,x\in\cD\,,\ v\in\cD^\perp\right\}\,,
\ee
which coincides with $V(L)$.
\subsection{Examples}
We now describe some examples of quantum dynamics generated by a Lagrangian or Hamiltonian, using
the quantum Tulczyjew triple with coordinates (\ref{coor}).
\begin{example}
Let $(\zl_k)_1^\infty$ be a sequence of non-zero real numbers. On $\sT H=H\oplus H$ with coordinates $(q,\dot q)$ consider the Lagrangian
\be\label{L1}
L=\frac{1}{2}\sum_{k=1}^\infty\left(\frac{1}{\zl_k}\left(\dot q^k\right)^2-\zl_k\left( q^k\right)^2\right)\,.
\ee
The Lagrangian is densely defined and its domain of differentiability is
\be\label{domain1}
\cD_0=\left\{(q,\dot q)\in\sT H\,:\,\sum_k|\zl_k q^k|^2<\infty\,,\quad \sum_k|\dot q^k/\zl_k|^2<\infty \right\}\,.
\ee
The Lagrangian submanifold it defines is $\xd L(\cD_0)$,
$$S(L)=\left\{\left(q^k,\dot q^k,-\zl_kq^k,\dot q^k/\zl_k\right)\,:\, (q,\dot q)\in\cD_0\right\}\subset\sT^*\sT H\,.
$$
The Lagrangian submanifold $V(L)=\za_H^{-1}(S(L))$ then reads
$$V(L)=\left\{\left(q^k,\dot q^k/\zl_k,\dot q^k,-\zl_kq^k\right)\,:\, (q,\dot q)\in\cD_0\right\}\subset\sT\sT^* H\simeq \cH\oplus\cH\,.
$$
It is easy to see that $V(L)$ is the graph of the complex linear operator $-\frac{i}{\hbar}A$, where
$$Ax=\hbar\sum_{k=1}^\infty\zl_kx^k\,,\quad x^k=(q^k+ip_k)e_k\,,$$
is self-adjoint on the domain
$$\cD=\left\{x\in\cH\,:\,\sum_{k=1}^\infty\nm{\zl_kx^k}^2<\infty\right\}\,.
$$
The corresponding Hamiltonian reads
\be\label{H1}
h_A(x)=\frac{1}{2}\sum_{k=1}^\infty{\zl_k}\nm{x^k}^2=\frac{1}{2}\sum_{k=1}^\infty{\zl_k}\left((p_k)^2+(q^k)^2\right)^2\,,
\ee
defined on its domain of differentiability $\cD$.
\end{example}

\begin{example}
Let us modify the above example by putting $\zl_1=0$ with the additional constraint on the domain of $L$:
\be\label{domain2}
\cD_0=\left\{(q,\dot q)\in\sT H\,:\,\dot q^1=0\,,\quad\sum_{k>1}|\zl_k q^k|^2<\infty\,,\quad \sum_{k>1}|\dot q^k/\zl_k|^2<\infty \right\}\,.
\ee
With this constraint our Lagrangian is:
\be\label{L2}
L=\frac{1}{2}\sum_{k=2}^\infty\left(\frac{1}{\zl_k}\left(\dot q^k\right)^2-\zl_k\left( q^k\right)^2\right)
\ee
and the corresponding (vakonomically generated) Lagrangian submanifold reads
$$S(L)=\left\{\left(q^k,\dot q^k,-\zl_kq^k,a^k\right)\,:\, (q,\dot q)\in\cD_0\right\}\subset\sT^*\sT H\,,
$$
where $a^k=\dot q^k/\zl_k$ if $k>1$, and $a^1$ is arbitrary.
The Lagrangian submanifold $V(L)=\za_H^{-1}(S(L))$ then reads
$$V(L)=\left\{\left(q^k,a^k,\dot q^k,-\zl_kq^k\right)\,:\, (q,\dot q)\in\cD_0\right\}\subset\sT\sT^* H\simeq \cH\oplus\cH\,.
$$
and is the graph of the complex linear operator $-\frac{i}{\hbar}A$, where
$$Ax=\hbar\sum_{k=2}^\infty\zl_kx^k\,,\quad x^k=(q^k+ip_k)e_k\,,$$
is self-adjoint on the domain
$$\cD=\left\{x\in\cH\,:\,\sum_{k=2}^\infty\nm{\zl_kx^k}^2<\infty\right\}\,.
$$
Now, $A$ has a non-trivial kernel spanned by $e_1$ and the range of $A$ is $\langle e_1\rangle^\perp$. The corresponding Hamiltonian reads
\be\label{H2}
h_A(x)=\frac{1}{2}\sum_{k=2}^\infty{\zl_k}\nm{x^k}^2\,.
\ee
\end{example}

\begin{example}
Let us modify again the first example, this time by putting $\zl_1=\infty$ with the additional constraint on the domain of $L$:
\be\label{domain3}
\cD_0=\left\{(q,\dot q)\in\sT H\,:\, q^1=0\,,\quad\sum_{k>1}|\zl_k q^k|^2<\infty\,,\quad \sum_{k>1}|\dot q^k/\zl_k|^2<\infty \right\}\,.
\ee
With this constraint our Lagrangian is again no longer completely regular and formally looks the same:
\be\label{L3}
L=\frac{1}{2}\sum_{k=2}^\infty\left(\frac{1}{\zl_k}\left(\dot q^k\right)^2-\zl_k\left( q^k\right)^2\right)\,.
\ee
However, the domain is now different and the corresponding (vakonomically generated) Lagrangian submanifold reads
$$S(L)=\left\{\left(q^k,\dot q^k,-b^k,a^k\right)\,:\, (q,\dot q)\in\cD_0\right\}\subset\sT^*\sT H\,,
$$
where $b^k=\zl_k q^k$ and $a^k=\dot q^k/\zl^k$ if $k>1$, and $b^1$ is arbitrary, $a^1=0$.
The Lagrangian submanifold $V(L)=\za_H^{-1}(S(L))$ then reads
$$V(L)=\left\{(x,y)\,:\, x\in\cD\,,\quad y^k=-ix^k\ \text{for}\ k>1\right\}\subset\sT\sT^* H\simeq \cH\oplus\cH\,.
$$
It is no longer the graph of a complex linear operator but it is a genuine anti self-adjoint relation in the domain
$$\cD=\left\{x\in\cH\,:\,\sum_{k=2}^\infty\nm{\zl_kx^k}^2<\infty\,,\quad x^1=0\right\}\,,
$$
which is not dense in $\cH$. The closure $\ol{\cD}$ is $\langle e_1\rangle^\perp$ and $V(L)$ defines an anti self-adjoint operator $(-\frac{i}{\hbar}A)$ on $\ol{\cD}$, where
$$Ax=\hbar\sum_{k=2}^\infty\zl_kx^k\,,\quad x^k=(q^k+ip_k)e_k\,.$$
The corresponding Hamiltonian is defined on $\cD$ and reads
\be\label{H3}
h_A(x)=\frac{1}{2}\sum_{k=2}^\infty{\zl_k}\nm{x^k}^2\,.
\ee

\end{example}

\begin{example}
Consider an open bounded domain $\zW$ in $\R^n$ and the Hilbert space $\cH=\L^2(\zW,\xd \br)$.
We have a canonical real part $H$ of $\cH$ consisting of real functions,
$H=\L^2_\R(\zW,\xd \br)$. For $U(\br)$ being a potential function on $\zW$, consider the `classical' Schr\"odinger operator $A=-\frac{\hbar^2}{2m}\zD + U$ and assume that it is self-adjoint in a domain $\cD_0$ and invertible. This operator can be reconstructed with the use of our procedure from the quadratic Lagrangian
\be\label{lag}
L(q,\dot q)=\frac{1}{2}\int_\zW\left[\dot q A^{-1}\dot q- \frac{1}{\hbar^2}qAq  \right]\xd \br\,.
\ee
This corresponds to the observations made in \cite{MV}. With the use of the Green function $G$ for $A$, we can write
\be\label{lag1}
L(q,\dot q)=\frac{1}{2}\int_\zW\left[\int_\zW\left(\dot q(\br)G(\br,\br')\dot q(\br')\right)\xd\br'- \frac{1}{\hbar^2}qAq  \right]\xd \br\,.
\ee
As a particular example consider the self-adjoint extension $A=-\frac{\hbar^2}{2m}\zD$ of the Laplace operator defined in $\zW=\R$ on smooth functions vanishing at infinity with first derivatives. This situation differs from the above examples, as $A$ has no point spectrum. Nevertheless, we can find the Lagrangian in the form (\ref{lag1}), which in this case reads
\be\label{lag2}
L(q,\dot q)=\frac{1}{2}\int_{\R}\left[\int_\R\left(\frac{2m}{\hbar^2}\dot q(\br)(\br-\br')\Theta(\br-\br')\dot q(\br')\right)\xd\br'- \frac{1}{2m}q(\br)(\zD q)(\br)  \right]\xd \br\,.
\ee
Here, $\Theta$ is the Heaviside step function ($\br\Theta(\br)$ is the Green function for $\zD$).

\end{example}

\subsection{Hamiltonian formalism on the Hilbert projective space}\label{holohamiltononians}
Since the Hamiltonians (\ref{H}) are constant along the action of the subgroup $S^1=\langle e^{it}\Id\rangle$ in $\U(\cH)$, we can carry out a Hamiltonian reduction from the unit sphere $S^\infty(\cH)$ which is an isotropic submanifold in $\cH$ onto the Hilbert projective space
$\P(\cH)$. The reduced symplectic structure is the canonical symplectic structure on $\P(\cH)$ and the reduced Hamiltonian is
\be\label{ha}h_{\P A}([\psi])=\frac{1}{2\hbar}\frac{\bk{\psi}{A\psi}}{\bk{\psi}{\psi}}\,.
\ee
which is differentiable on the domain $\cD_\P=\{[\psi]\,|\,\psi\in\cD\}$, where $\cD$ is the domain of the self-adjoint operator $A$ on the Hilbert space $\ol{\cD}$. The corresponding (implicit) Hamiltonian dynamics is represented by the Lagrangian submanifold
\be\label{rhLagrangian}V(h_{\P A})=\{\phi_\psi\,|\,[\psi]\in\cD_\P\,,\quad\phi=-\frac{i}{\hbar\nm{\psi}^2}(A\psi)^\perp+v\,,\quad v\in\langle\psi\rangle^\perp\cap\ol{\cD}^\perp\}\,,
\ee
in $\sT\P(\cH)$, where
$$(A\psi)^\perp=A\psi-\frac{\bk{\psi}{A\psi}}{\bk{\psi}{\psi}}\psi$$
is the part of $A\psi$ orthogonal to $\psi$ in the decomposition $\cH=\langle\psi\rangle\op\langle\psi\rangle^\perp$ (the definition is correct: we should remember that $(\ol{\za}\phi)_\psi=\phi_{(\za\psi)}$). It is easy to see that, in other words,
$V(h_{\P A})$ is the projection of $V(h_A)$ under the tangent map $\sT\P$ of the canonical projection $\P:\cH^\ti\ni\psi\to[\psi]\in\P(\cH)$ (cf.\    (\ref{Vh})).
Reducing eventually to the Hilbert projective space $\ol{\cD_\P}=\P(\ol{\cD})$, we can assume that
$\cD_\P$ is dense in $\P(\cH)$, i.e.\   $h_{\P A}$, thus $A$ and the Hamiltonian vector fields
$X_{A}(x)=-\frac{i}{2\hbar}Ax$, as well as
$$X_{\P A}([\psi])=-\frac{i}{\hbar\nm{\psi}^2}((A\psi)^\perp)_\psi\,,
$$
are densely defined. The vector field $X_{\P A}$ is an infinitesimal automorphism of the K\"ahler structure, since it generates a one-parameter group of automorphism continuous in the strong topology on $\PU(\cH)$ (i.e.\   the topology induced from the strong topology on the unitary group), namely the one given by
\be\label{var}\R\ni t\mapsto \P_{\exp(-itA/\hbar)}\in\PU(\cH)\,.
\ee
Densely defined Hamiltonian vector fields on $\P(\cH)$ will be called \emph{real holomorphic} if they are infinitesimal automorphisms of the K\"ahler structure (cf.\    Remark
\ref{re1}), and the corresponding Hamiltonians we call \emph{real holomorphic Hamiltonians}.

The following theorem is a nontrivial but very useful fact.
\begin{theorem}\label{pself} Each strongly  continuous group of automorphisms $[U]_t$ of the K\"ahler structure on $\P(\cH)$ is of this form (\ref{var}). In other words, each real holomorphic
Hamiltonian is of the form (\ref{ha}) for a self-adjoint operator $A$ on $\cH$.
\end{theorem}
\begin{proof} As every $[U]_t$ has a representative $U_t$ in the unitary group, $[U]_t=\P_{U_t}$, one shows first that these representatives can be chosen continuously, i.e.\   such that the map $\R\ni t\mapsto U_t\in\U(\cH)$ from $\R$ into the unitary group with the strong topology is continuous (see \cite[Chapter 7]{Var}). As $[U]_t$ is a one-parameter group there is a \emph{$S^1$-multiplier} $m:\R\ti\R\to S^1\subset\C$ such that
$$U_t\circ U_{t'}=m(t,t')U_{t+t'}\,.$$
If we choose other representatives, then the resulted multiplier $m'$ is \emph{similar} to $m$, i.e.\   there exists a function $s:\R\to S^1$ such that
$$m'(t,t')=\frac{m(t,t')s(t+t')}{s(t)s(t')}\,.
$$
A multiplier which is similar to the trivial multiplier $m_0(t,t')=1$ we call \emph{exact}.
Finally, one proves that any $S^1$-multiplier on $\R$ is exact \cite[Theorem 7.38]{Var}, so we can choose unitary representatives which form a strongly continuous one-parameter group of unitary transformations.

\end{proof}

\medskip
The Lagrangian submanifolds (\ref{rhLagrangian}) in $\sT\P(\cH)$ we will call \emph{anti self-adjoint dynamics}. We have the following characterization which easily follows from (\ref{projection}) and Theorem \ref{pself}.
\begin{theorem}
Anti self-adjoint dynamics in $\sT\P(\cH)$ are exactly projections under $\sT\P:\sT\cH=\cH\oplus\cH\to\sT\P(\cH)$ of anti self-adjoint relations in $\sT\cH$, i.e.
of the form (\ref{rhLagrangian}) for a certain self-adjoint operator $A:\cH\supset\cD\to\cH$
densely defined on the Hilbert space $\ol{\cD}$. They correspond to constrained real holomorphic Hamiltonians (\ref{ha}).
\end{theorem}

Assume now for simplicity that a self adjoint operator $A$ is densely defined in $\cH$. The following is well known \cite{CM}.
\begin{theorem}
The critical point of the Hamiltonian (\ref{ha}) correspond to eigenvectors of $A$ and the values
of the Hamiltonian at the critical values equal the corresponding eigenvalues up to the factor $\frac{1}{2\hbar}$.
\end{theorem}
\begin{proof}
According to (\ref{rhLagrangian}), the differential of the Hamiltonian vanishes at $[\psi]$ if and only if $(A\psi)^\perp=0$, i.e.\   $A\psi=\lambda\psi$ for some $\lambda$ (which must be real).
In this case, $h_{\P A}([\psi])=\frac{\zl}{2\hbar}$.
\end{proof}

\medskip
Note that, from the point of view of dynamics on the Hilbert projective space, the operator $A$ is determined only up to a component $\zl\Id$ for some $\zl\in\R$, so that
$$\P_{\exp(tA)}=\P_{\exp(tA')}\ \Leftrightarrow\ A-A'=\zl\Id\,,\ \zl\in\R\,.
$$
Hence, the set
$$\A([U]_t)=\{ A\,|\, \P_{\exp(-itA/\hbar)}=[U]_t\}$$
is an affine line. The corresponding spectrum of $\A([U]_t)$ must be understood as an affine object, i.e.\   a subset in the affine real line. Fixing $A\in\A([U]_t)$ identifies $\A([U]_t)$ with $\R$ and $\zs(\A([U]_t))$ with $\zs(A)\subset\R$. The definition is correct, since adding $\zl\Id$ to $A$ results in shifting the spectrum of $A$ by $\zl$.
In other words,
$$\zs(\A([U]_t))=\{ A\in\A([U]_t)\,|\, A^{-1}\in\gl(\cH)\}\,.
$$

\section{Self-adjoint extensions of symmetric relations}
Recall that on $\cH_1=\cH\oplus\cH$ we have considered the pseudo-Hermitian products
\beas
\bk{(x,\dot x)}{(y,\dot y)}_{0-}&=&i\left(\bk{\dot x}{y}-\bk{x}{\dot y}\right)\,,\\
\bk{(x,\dot x)}{(y,\dot y)}_{0+}&=&\bk{\dot x}{y}+\bk{x}{\dot y}\,,\\
\bk{(\zf,\zf')}{(\zc,\zc')}_\pm &=&\bk{\zf'}{\zc'}\pm\bk{\zf}{\zc}\,,
\eeas
and the corresponding symplectic forms $\zw_{0\pm}$, $\zw_\pm$ being the imaginary parts of the above products. We know that (complex) linear Lagrangian subspaces of $\zw_{0-}$ and $\zw_{0+}$ correspond to self-adjoint and anti self-adjoint relations (operators if the relations are graphs), respectively.
As for $\zw_{\pm}$ we have the following.
\begin{proposition}\label{prop1}
Complex linear isotropic submanifolds of $\zw_{-}$ are always graphs of a (partially defined)
complex linear operator $U:\cH\subset\cD\rightarrow\cH$ which is a partial isometry. They are Lagrangian if and only if $U$ is unitary.
\end{proposition}
\begin{proof}
If a complex linear subspace $V\subset\cH_1$ is isotropic with respect to $\zw_{-}$, then by the standard argument it is isotropic with respect to $\bk{\cdot}{\cdot}_{-}$. Hence, if $(\zf,\zf'_j)\in V$, $j=1,2$, then $(0,\zf'_1-\zf'_2)\in V$, so
$$\bk{(0,\zf'_1-\zf'_2)}{(0,\zf'_1-\zf'_2)}=\nm{\zf'_1-\zf'_2}^2=0$$
and $\zf'_1=\zf'_2$. This shows that $V$ is a graph, $V=\mathfrak{G}(U)$. Now, the isotropy property for the graph of $U$ reads
$$\bk{Ux}{Uy}=\bk{x}{y}\,,$$
so $U$ is a partial isometry. Moreover, $\mathfrak{G}(U)$ is Lagrangian implies that the domain of $U$ cannot be enlarged, thus is the whole $\cH$. Since the transposition $\cH_1\ni(\zf,\zf')\mapsto(\zf',\zf)\in\cH$ is an anti-symplectomorphism for $\zw_-$, it maps Lagrangian subspaces onto Lagrangian subspaces. Therefore the image of $U$ must also be the whole $\cH$, i.e.\   $U$ is invertible and $U\in\U(\cH)$.
\end{proof}
\begin{remark}
We have a similar result for the symplectic structure and complex anti-linear Lagrangian submanifolds which turn out to be the graphs of anti-unitary operators.
\end{remark}

Consider the complex linear isomorphism $C:\cH\oplus\cH\to \cH\oplus\cH$,
\be\label{cay}
C(x,x')=\frac{1}{\sqrt{2}}\left(x'+ix,x'-ix\right)\,.
\ee
It is easy to see that $C$ transfers $\bk{\cdot}{\cdot}_+$ into $\bk{\cdot}{\cdot}_{0+}$ and
$\bk{\cdot}{\cdot}_{0-}$ into $\bk{\cdot}{\cdot}_-$. Indeed,
\beas
\bk{C(\zf,\zf')}{C(\zc,\zc')}_{0+}&=&\frac{1}{{2}}\bk{(\zf'-i\zf)}{(\zc'+i\zc)}
+\frac{1}{{2}}\bk{(\zf'-i\zf)}{(\zc'+i\zc)}\\
&=&\bk{\zf'}{\zc'}+\bk{\zf}{\zc}=\bk{(\zf,\zf')}{(\zc,\zc')}_+\,.
\eeas
Similarly,
\beas
\bk{C(x,\dot x)}{C(y,\dot y)}_{-}&=&\frac{1}{{2}}\bk{(\dot x-ix)}{(\dot y-iy)}
-\frac{1}{{2}}\bk{(\dot x+ix)}{(\dot y+iy)}\\
&=&i\left(\bk{x}{\dot y}+\bk{x}{\dot y}\right)=\bk{(x,\dot x)}{(y,\dot y)}_{0-}\,.
\eeas
We have the following.
\begin{theorem}
The map (\ref{cay}) is, up to a unitary map, the unique complex linear isomorphism $C:\cH\oplus\cH\to \cH\oplus\cH$ which identifies the symplectic structures $\zw_{0+}$ and $\zw_{+}$.
\end{theorem}
\begin{proof}
If $C'$ is another such an isomorphism, then $C'C^{-1}$ is a complex linear isomorphism preserving
$\zw_+$, thus by the general argument, preserving $\bk{\cdot}{\cdot}_+$, so it is unitary on $\cH_1$.
\end{proof}

Note that the symplectomorphism $C$ can be interpreted as the complex \emph{Cayley transform} and we easily obtain the von Neumann's characterization of self-adjoint extensions of symmetric operators in terms of \emph{deficiency spaces}.
\begin{theorem}
If $V\subset\cH\oplus\cH$ is a symmetric relation, i.e.\   an isotropic submanifold with respect to the symplectic form $\zw_{0-}$, then its image $C(V)\subset\cH\oplus\cH$ by $C$ is the graph of a partial isometry
$$U:\cH\supset W_+\to W_-\subset \cH\,,$$
where $W_+$, $W_-$ are canonical projections of $C(V)\subset\cH\oplus\cH$ onto the first and the second component, respectively.
The set of self-adjoint relations extending $V$ can be identified with the set of unitary operators
$$U_0:N_+\to N_-\,,$$
where $N_\pm$ are the \emph{deficiency spaces},
$$N_+=W_+^\perp\,,\quad N_-=W_-^\perp\,.$$
\end{theorem}
\begin{proof}
$C(V)$ is isotropic with respect to $\zw_-$ and, according to Proposition \ref{prop1}, $C(V)$ is the graph $\mathfrak{G}(U)$ of a partial isometry $U$.
Moreover, any self-adjoint extension of $V$ corresponds to a complex linear Lagrangian submanifold $V'$ containing $V$. Its image $C(V')$ under $C$ is complex linear Lagrangian for $\zw_-$, thus the graph of an unitary operator $U'$ extending $U$, thus unitarily identifying the orthogonal complements of the domain $W_+$ and the range $W_-$ of $U$.
\end{proof}

\begin{corollary} (von Neumann)
If $A$ is a symmetric operator on $\cH$, then the image $C(\mathfrak{G}(A))\subset\cH\oplus\cH$ by $C$ of its graph $\mathfrak{G}(A)\subset\cH\oplus\cH$ is the graph of a partial isometry
$$U:\cH\supset W_+(A)\to W_-(A)\subset \cH\,,$$ which satisfies
\be\label{cf}
U(A+iI)=A-iI\,.
\ee
Here, $W_\pm(A)$ is the range of the operator $A\pm iI$. Complements of the ranges can be identified with kernels of adjoint operators.
The set of self-adjoint extensions of $A$ can be therefore identified with the set of unitary operators
$$U_0:N_+(A)\to N_-(A)\,,$$
where $N_\pm(A)$ are the \emph{deficiency spaces} being the kernels of $A^\dagger\mp iI$,
$$N_\pm(A)=\ker(A^\dag\mp iI)\,.$$
\end{corollary}
\begin{proof}
We know that $\mathfrak{G}(A)$ is isotropic for symmetric $A$. Hence, $C(\mathfrak{G}(A))$ is the graph $\mathfrak{G}(U)$ of a partial isometry $U$. Since elements of $\mathfrak{G}(A)$ are of the form $(x,Ax)$ and
$$C(x,Ax)=\frac{1}{\sqrt{2}}\left(Ax+ix,Ax-ix\right)\,,
$$
the element $C(x,Ax)$ is in the graph of $U$ if and only if
$$U(Ax+ix)=Ax-ix\,.$$
Moreover, any self-adjoint extension of $A$ corresponds to the graph of an unitary operator $U'$ extending $U$, thus unitarily identifying the orthogonal complements of the domain and the range of $U$, thus the ranges of $A\pm iI$. Clearly they are the kernels of the operators $A^\dag\mp iI$.

\end{proof}
\section{Additional questions}

\subsection{The Heisenberg picture}
Within the standard formulation of quantum mechanics, the Schr\"odinger picture is considered to be equivalent to the Heisenberg picture. It is therefore meaningful to ask to what extent our geometrical formulation of the dynamics in the Schr\"odinger picture has a counterpart also in the Heisenberg picture. In this section we would like to outline a possible approach to deal with Heisenberg equations of motion within the geometrical formulation.

If $A$ and $T$ are quantum observables (self-adjoint operators), then the one-parametr group of unitary transformations $U_t=\exp(-itA)$ produces the trajectory of self-adjoint operators (observables) $T_t=U_tTU_t^\dag$. It is however difficult to describe the `generator' of such one-parameter group of transformations; it would be in general an unbounded `derivation'.

The situation is better when we reduce to bounded $T$. In this case, the generator is usually understood as the commutator $i[A,T]=i(AT-TA)$. However, in general this commutator makes no sense even for all bounded $T$ and is defined as an unbounded operator on the Banach space $\gl(\cH)$ for a dense set of $T$. On the other hand, we can develop our full machinery when we reduce to the space $\L^2(\cH)$ of Hilbert-Schmidt operators. Since Hilbert-Schmidt operators can be interpreted as elements of the Hilbert tensor product $\cH\otimes{\cH}^*$, where the dual $\cH^*$ can be identified with $\cH$ with respect to the anti-linear isomorphism (cf. \cite{grabowski07})
$$\cH\ni x\mapsto \ol{x}\in\cH^*\,,$$
understood as identifying `bras' with `kets', $\bra{x}\mapsto\ket{x}$, we can interpret the commutator as an action in the tensor product:
$$[A,x\otimes\ol{y}]=A\circ x\otimes\ol{y}-x\otimes\ol{y}\circ A=Ax\otimes\ol{y}-x\otimes\overline{Ay}\,.$$
Then, we can develop the whole machinery for the corresponding quantum dynamics just replacing $\cH$ with $\cH\otimes\cH$.
\subsection{Composite systems}
In addition to states, observables, probability functions and evolution, a basic requirement for the description of quantum systems is a composition rule, i.e., how to compose interacting quantum systems. Therefore our description should take into account also the composition rule for quantum systems.
One of the main difference between classical and quantum mechanics is that to obtain composition of two systems one uses the Cartesian product of configuration spaces in the classical case, and the tensor product of the corresponding Hilbert spaces in the quantum case. The dimensions of the tensor products are much higher than the dimension of the Cartesian product, this fact is usually understood to be the source of quantum phenomena like entanglement.

How we can explain all this in our model with the Tulczyjew triple, being originally classical.
Note first that our configuration spaces $H$ are by definition linear (real Hilbert spaces).
According to the classical rule, for the composition of systems with configurations in $H_1$, $H_2$, respectively, we should use $H_1\ti H_2$. But our configurations should be linear, so we must generate freely (if we do not want to introduce extra constraints) a real Hilbert space out of $H_1\ti H_2$. This is exactly the (real) tensor product $H_1\otimes_\R H_2$. Our complex Hilbert space is therefore the complexification of $H_1\otimes_\R H_2$ which, as easily seen, is the complex tensor product $\cH_1\otimes_\C\cH_2$ of the complexifications of $H_1$ and $H_2$.
\section{Concluding remarks}
We have reviewed the essential aspects of the Lagrangian description of implicit differential equations on the cotangent bundle of a configuration space by means of the Tulczyjew triple. By taking advantage of the geometrical formulation of quantum mechanics we have proposed a Lagrangian description , similar to the classical one, also for quantum dynamics.
This manifold point of view has required revisiting of various aspects in the framework of differential geometry applied to Hilbert manifolds.

The characterization in terms of Lagrangian submanifolds (relations) allows to reformulate the problem of selfadjointness for unbounded operators in geometrical terms. The formalism we have considered appears to be quite flexible, therefore in a forthcoming paper we shall consider composite systems and discuss the problem of separability and entanglement by `geometrizing' our previous approach \cite{grabowski07,grabowski11}

\section*{Acknowledgments}
We are indebted to Karl-Hermann Neeb for his useful comments on the unitary group. G.~Marmo would like to thank the Santander -- UCIIIM Chair of Excellence Program 2016/2017 for partial support. J.~Grabowski and T.~Shulman acknowledge the support of the  Polish National Science Centre grant under the contract number DEC-2012/06/A/ST1/00256.

\noindent Janusz Grabowski\\\emph{Institute of Mathematics, Polish Academy of Sciences}\\{\small \'Sniadeckich 8, 00-656 Warszawa,
Poland}\\{\tt jagrab@impan.pl}
\\

\noindent Marek Ku\'s\\
\emph{Center for Theoretical Physics, Polish Academy of Sciences,} \\
{\small Aleja Lotnik{\'o}w 32/46, 02-668 Warszawa,
Poland} \\{\tt marek.kus@cft.edu.pl}
\\

\noindent Giuseppe Marmo\\
\emph{Dipartimento di Fisica, Universit\`{a} ``Federico II'' di Napoli} \\
\emph{and Istituto Nazionale di Fisica Nucleare, Sezione di Napoli,} \\
{\small Complesso Universitario di Monte Sant Angelo,} \\
{\small Via Cintia, I-80126 Napoli, Italy} \\
{\tt marmo@na.infn.it}
\\

\noindent Tatiana Shulman\\\emph{Institute of Mathematics, Polish Academy of Sciences}\\{\small \'Sniadeckich 8, 00-656 Warszawa,
Poland}\\{\tt tshulman@impan.pl}


\end{document}